\documentclass[a4paper,USenglish]{class}
 
\usepackage{microtype}
\usepackage{stmaryrd,amssymb, amsfonts, amsthm}
\usepackage{mathdots}
\usepackage{pgf}
\usepackage{manfnt}
\usepackage{color}
\definecolor{lavaflower}{HTML}{710F12}
\definecolor{regalblue}{HTML}{002366}
\usepackage{hyperref}
\hypersetup{colorlinks=true,citecolor=red,linkcolor=blue}		
\usepackage{thm-restate}
\usepackage{enumitem}
\usepackage{tikz}
\usetikzlibrary{calc,arrows,shapes,backgrounds}
\usepackage{bwc}

\title{Threshold Constraints with Guarantees for Parity Objectives in Markov Decision Processes\footnote{Work partially supported by the ERC Starting grant 279499 (inVEST) and the ARC project ``Non-Zero Sum Game Graphs: Applications to Reactive Synthesis and Beyond'' (F\'ed\'eration Wallonie-Bruxelles). J.-F. Raskin is Professeur Francqui de Recherche, M.~Randour is an F.R.S.-FNRS postdoctoral researcher.}}
\titlerunning{Threshold Constraints with Guarantees for Parity Objectives in MDPs} 
\author[1,2]{Rapha\"el Berthon}
\author[2]{Mickael Randour}
\author[2]{Jean-Fran\c{c}ois Raskin}
\affil[1]{ENS Rennes, France}
\affil[2]{Computer Science Department, ULB - Universit\'e libre de Bruxelles, Belgium}
\authorrunning{R.~Berthon, M.~Randour, and J.-F.~Raskin}

\Copyright{Rapha\"el Berthon, Mickael Randour, and Jean-Fran\c{c}ois Raskin}

\begin{document}
\maketitle

\begin{abstract}
The {\em beyond worst-case synthesis problem} was introduced recently by Bruy\`ere et al.~\cite{DBLP:conf/stacs/BruyereFRR14}: it aims at building system controllers that provide strict worst-case performance guarantees against an antagonistic environment while ensuring higher expected performance against a stochastic model of the environment. Our work extends the framework of~\cite{DBLP:conf/stacs/BruyereFRR14} and follow-up papers, which focused on quantitative objectives, by addressing the case of $\omega$-regular conditions encoded as parity objectives, a natural way to represent functional requirements of systems.

We focus on building strategies that satisfy a main parity objective on all plays, while ensuring a secondary one with sufficient probability. This setting raises new challenges in comparison to quantitative objectives, as one cannot easily mix different strategies without endangering the functional properties of the system. Interestingly, we establish that, for all variants of this pro\-blem, deciding the existence of a strategy  lies in ${\sf NP} \cap {\sf coNP}$, the same complexity class as classical parity games. Hence, our framework provides additional modeling power while staying in the same complexity class.
\end{abstract}

\section{Introduction}

\smallskip\noindent\textbf{Beyond worst-case synthesis.} \textit{Two-player zero-sum games}~\cite{DBLP:conf/dagstuhl/2001automata,rECCS} and \textit{Markov decision processes} (MDPs)~\cite{filar1997,baier2008principles} are two popular frameworks to model decision making in adversarial and uncertain environments respectively. In the former, a system controller (player~1) and its environment (player~2) compete antagonistically, and synthesis aims at building strategies for the controller that ensure a specified behavior \textit{against all possible strategies of the environment}. In the latter, the system is faced with a given stochastic model of its environment, and the focus is on satisfying a given level of expected performance, or a \textit{specified behavior with a sufficient probability}.

The \textit{beyond worst-case synthesis} framework was introduced by Bruy\`ere et al.~\cite{DBLP:conf/stacs/BruyereFRR14} to unite both views: in this setting, we look for strategies that provide both \textit{strict worst-case guarantees} and a \textit{good level of performance against the stochastic model}. Such requirements are natural in many practical situations (e.g., see~\cite{DBLP:journals/corr/BruyereFRR14,DBLP:conf/vmcai/RandourRS15} for applications to the shortest path problem). The original paper~\cite{DBLP:conf/stacs/BruyereFRR14} dealt with mean-payoff and shortest path quantitative settings. Substantial follow-up work include, e.g., multi-dimensional extensions~\cite{clemente2015multidimensional}, optimization of the expected mean-payoff under hard Boolean constraints~\cite{DBLP:conf/concur/AlmagorKV16} or under energy constraints~\cite{DBLP:conf/atva/BrazdilKN16}, or integration of beyond worst-case concepts in the tool \textsc{Uppaal}~\cite{DBLP:conf/atva/DavidJLLLST14}.

\smallskip\noindent\textbf{Parity objectives.} In this paper, we study the beyond worst-case problem for $\omega$-regular conditions encoded as \textit{parity objectives}. Parity games have been under close scrutiny for a long time both due to their importance (e.g., they subsume modal $\mu$-calculus model checking~\cite{DBLP:conf/cav/EmersonJS93}) and their intriguing complexity: they belong to the class of problems in $\sf NP \cap coNP$~\cite{DBLP:journals/ipl/Jurdzinski98} and despite many efforts (see~\cite{DBLP:journals/corr/BruyereHR16} for pointers), whether they belong to $\sf P$ is still an open question. 

In the aforementioned papers dealing with beyond worst-case problems, the focus was on \textit{quantitative} objectives (e.g., mean-payoff). While it is usually the case that \textit{qualitative} objectives, such as parity, are easier to deal with than quantitative ones, this is not true in the setting considered in this paper. Indeed, in the context of quantitative objectives, it is conceivable to alternate between two strategies along a play, such that one -- efficient -- strategy balances the performance loss due to playing the other -- less efficient -- strategy for a limited stretch of play infinitely often. In the context of qualitative objectives, this is no more possible in general, as one strategy may induce behaviors (such as invalidating the parity condition infinitely often) that can never be counteracted by the other one. Hence, in comparison, we need to define more elaborate analysis techniques to detect when satisfying \textit{both} the worst-case and the probabilistic constraints with a single strategy is actually possible.

\begin{figure}[tbh]
\centering
\scalebox{0.8}{\begin{tikzpicture}[every node/.style={font=\small,inner sep=1pt}]
\draw (0,0) node[rond,bleu] (s0) {$1,0$};
\draw (-2,1) node[carre,bleu] (s1) {$1,0$};
\draw (-2,-1) node[carre,vert] (s2) {$2,0$};
\draw (2,0) node[carre,jaune] (s3) {$2,1$};
\draw (0,-0.6) node (l0) {$a$};
\draw (-2.5,1) node (l1) {$b$};
\draw (-2.5,-1) node (l2) {$c$};
\draw (2,-0.6) node (l3) {$d$};
\draw[-latex] (s1) to (s2);
\draw[-latex] (s2) to (s0);
\draw[-latex] (s0) to[out=100,in=355] (s1);
\draw[-latex] (s1) to[out=290,in=175] (s0);
\draw[-latex] (s0) to[out=30,in=150] (s3);
\draw[-latex] (s3) to[out=210,in=330] (s0);
\end{tikzpicture}}
\caption{An MDP where player~1 can ensure $p_1$ surely and $p_2$ almost-surely.}
\label{fig:example}
\end{figure}
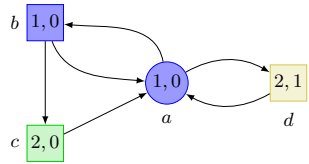

\begin{example} 
Consider the MDP of Figure~\ref{fig:example}. Circle states are owned by player~1 (system) and square states are owned by player~2 (environment). In the stochastic model of the environment, square states are probabilistic, and, when not specified, we consider the uniform distribution over their successors. Each state is labelled with a name and two integers $x, y$ representing priorities defined by two functions, $p_1$ and $p_2$.
An infinite path in the graph is winning for player~1 and parity objective $p_i$, $i \in \{1,2\}$, if the \textit{maximal} priority seen infinitely often along the path for function $p_i$ is {\em even}. We claim that player~1 has a strategy $\lambda$ to ensure that (i) all plays consistent with $\lambda$ satisfy $p_1$ (i.e., $p_1$ is surely satisfied) and (ii) the probability measure induced by $\lambda$ on this MDP ensures that $p_2$ is satisfied with probability one (i.e., almost-surely).

One such $\lambda$ is as follows. It plays an infinite sequence of rounds of $n_i$ steps, $i \in \mathbb{N}$. In round $i$, in state $a$, the strategy chooses $b$ for $n_i$ steps, such that the probability to reach~$c$ during round $i$ is larger than $1-2^{-i}$ (this is possible as at each step $c$ is reached from $b$ with probability $\frac{1}{2}$). If during round $i$, $c$ is not reached (which can happen with a small probability) then $\lambda$ goes to $d$ once. Then the next round $i+1$ is started. This infinite-memory strategy ensures both (i) and (ii). Indeed, it can be shown that the probability that $\lambda$ plays~$d$ infinitely often is zero. Also, during each round, the maximal priority for $p_1$ is guaranteed to be even because if $c$ is not visited, $d$ is systematically played.

Finally, we can prove that player~1 needs infinite memory to ensure $p_1$ surely and $p_2$ almost-surely, and also, that this is the best that player~1 can do here: he has no strategy to enforce surely both $p_1$ and $p_2$ at the same time.
\end{example}

\smallskip\noindent\textbf{Outline and contributions.} This paper extends a preceding conference version~\cite{icalp2017}. We consider MDPs with two parity objectives (i.e., using different priority functions). We study the problem of deciding the existence of a strategy that ensures the first parity objective \textit{surely} (i.e., on all plays) while yielding a \textit{probability at least equal to (resp.~greater than) a given rational threshold} to satisfy the second parity objective. In Section~\ref{sec:prelim}, we formally define the framework and recall important results from the literature. In Section~\ref{sec:reach}, as an intermediate step, we solve the problem of ensuring the first parity objective surely while visiting a target set of states with sufficient probability: this tool will help us several times later. We prove that the corresponding decision problem is in ${\sf NP} \cap {\sf coNP}$ and at least as hard as parity games, and that finite-memory strategies are sufficient. In Section~\ref{sec:asp}, we solve the problem for the two parity objectives, where the second one must hold \textit{almost-surely} (i.e., with probability one). Our main tools are the novel notion of \textit{ultra-good end-components}, as well as the reachability problem solved in Section~\ref{sec:reach}. We generalize our approach to arbitrary probability thresholds in Section~\ref{sec:threshp}, in which we introduce the notion of \textit{very-good end-components}. In both the almost-sure and the arbitrary threshold cases, we prove that the decision problem belongs to ${\sf NP} \cap {\sf coNP}$ and is at least as hard as parity games. In contrast to the reachability case, we prove that infinite memory is in general necessary.

\smallskip\noindent\textbf{Additional related work.} The beyond worst-case synthesis framework is an example of the usefulness of non-zero-sum games for reactive synthesis. More can be found in these expository papers~\cite{DBLP:conf/lata/BrenguierCHPRRS16,DBLP:journals/corr/Randour16}. Other types of multi-objective specifications in stochastic models have been considered in the literature: e.g., percentile queries generalize the classical threshold probability problem to several dimensions~\cite{percentile2017}.

In~\cite{DBLP:conf/atva/BaierGC09}, Baier et al. study the quantitative analysis of MDPs under weak and strong fairness constraints. They provide algorithms for computing the probability for $\omega$-regular properties in worst and best-case scenarios, when considering strategies that in addition satisfy weak or strong fairness constraints almost-surely. In contrast, we are able to consider similar objectives but for strategies that satisfy weak or strong fairness constraints surely, i.e., with certainty and not only with probability one.

In~\cite{DBLP:conf/concur/AlmagorKV16}, Almagor et al.~consider the optimization of the expected mean-payoff under hard Boolean constraints in weighted MDPs. Our concept of ultra-good end-component builds upon their notion of super-good end-component. A reduction to mean-payoff parity games~\cite{DBLP:conf/lics/ChatterjeeHJ05} is part of the identification process of both types of end-components.
\section{Preliminaries}
\label{sec:prelim}

\subsection{Core concepts}

\smallskip\noindent\textbf{Directed graphs.} A \emph{directed graph} is a pair $G=(S,E)$ with $S$ a set of vertices, called \textit{states}, and $E\subseteq S\times S$ a set of directed edges. We focus here on \textit{finite} graphs (i.e., $\vert S \vert < \infty$).
Given a state $s\in S$, we denote by $\mathtt{Succ}(s)=\{s'\in S \mid (s,s')\in E\}$ the set of successors of $s$ by edges in $E$. 
We assume that graphs are non-blocking, i.e., for all $s\in S$, $\mathtt{Succ}(s)\neq \emptyset$.

A \textit{play} in $G$ from an initial state $s \in S$ is an infinite sequence of states $\pi = s_0 s_1 s_2 \dotso$ such that $s_0=s$ and $(s_i,s_{i+1})\in E$ for all $i \geq 0$. The \textit{prefix} up to the $(n+1)$-th state of $\pi$ is the finite sequence $\pi(0,n)=s_0 s_1 \dotso s_n$.
We resp. denote the first and last states of a prefix $\rho = s_0 s_1 \dotso s_n$ by $\mathtt{First}(\rho)=s_0$ and $\mathtt{Last}(\rho) = s_n$.
For a play $\pi$, we naturally extend the notation to $\mathtt{First}(\pi)$.
Finally, for $i\in\mathbb{N}$, $\pi(i) = s_i$, and for $j > i$, $\pi(i,j) = s_i \dotso s_j$. 
The set of plays of $G$ is denoted by $\mathtt{Plays}(G)$. 
The corresponding set of prefixes is denoted by $\mathtt{Pref}(G)$. For a set of plays $\Pi$, we will also denote by $\mathtt{Pref}(\Pi)$ the set of all prefixes of these plays. Given two prefixes $\rho = s_0 \dotso s_m$ and $\rho' = s'_0 \dotso s'_n$ in $\mathtt{Pref}(G)$, we denote their \textit{concatenation} as $\rho \cdot \rho' = s_0 \dotso s_m s'_0 \dotso s'_n$. This concatenation is not necessarily a valid prefix of $G$. The same holds for a prefix concatenated with a play.

\smallskip\noindent\textbf{Probability distributions.} Given a countable set $A$, a (rational) \textit{probability distribution} on $A$ is a function $p\colon A \rightarrow [0, 1] \cap \mathbb{Q}$ such that $\sum_{a \in A} p(a) = 1$. We denote the set of probability distributions on $A$ by $\mathcal{D}(A)$. The \textit{support} of the probability distribution $p$ on $A$ is $\mathtt{Supp}(p) = \{a \in A \mid p(a) >0\}$.

\smallskip\noindent\textbf{Markov decision processes.} A (finite) \emph{Markov decision process} (MDP) is a tuple $\mymdp = (G, S_1, S_2, \delta)$ where (i) $G=(S,E)$ is a directed graph; (ii) $(S_{1},S_{2})$ is a partition of $S$ into states of player $1$ (denoted by $\playerOne$ and representing the system) and states of player~2 (denoted by $\playerTwo$ and representing the \textit{stochastic} environment); (iii) $\delta\colon S_2 \rightarrow \mathcal{D}(S)$ is the transition function that, given a stochastic state $s\in S_{2}$, defines the probability distribution $\delta(s)$ over the successors of $s$, such that for all $s \in S_2$, $\mathtt{Supp}(\delta(s)) = 
\mathtt{Succ}(s)$ (i.e., all outgoing edges of have non-zero probability to be taken). An MDP where for all $s \in S_1$, $\vert \mathtt{Succ}(s)\vert = 1$ is a fully-stochastic process called a \textit{Markov chain} (MC).

A prefix $\rho \in \mathtt{Pref}(\mymdp)$ belongs to $\player{i}$, $i \in \{1, 2\}$, if $\mathtt{Last}(\rho) \in S_i$. The set of prefixes that belong to $\player{i}$ is denoted by $\mathtt{Pref}_{i}(\mymdp)$.

\smallskip\noindent\textbf{Strategies.} A \textit{strategy} for $\playerOne$ is a function $\lambda\colon \mathtt{Pref}_1(\mymdp) \rightarrow \mathcal{D}(S)$, such that for all $\rho \in \mathtt{Pref}_1(\mymdp)$, we have $\mathtt{Supp}(\lambda(\rho)) \subseteq \mathtt{Succ}(\mathtt{Last}(\rho))$. The set of all strategies in $\mymdp$ is denoted by $\Lambda$. \textit{Pure} strategies have their support equal to a singleton for all prefixes. When a strategy $\lambda$ is pure, we simplify its notation and write $\lambda(\rho) = s$ instead of $\lambda(\rho)(s) = 1$, for any $\rho \in \mathtt{Pref}_1(\mymdp)$ and the unique state $s \in \mathtt{Supp}(\lambda(\rho))$. We sometimes mention that a strategy is \textit{randomized} to stress on the need for randomness in general (i.e., when pure strategies do not suffice).

A strategy $\lambda$ for $\playerOne$ can be encoded by a stochastic state machine with outputs, called \emph{stochastic Moore machine}, $\mathtt{M} = (M, m_0, \alpha_u, \alpha_n)$ where (i) $M$ is a finite or infinite set of memory elements, (ii) $m_0\in M$ is the initial memory element, (iii) $\alpha_u\colon M\times S\rightarrow M$ is the update function, and (iv) $\alpha_n\colon M \times S_1 \rightarrow \mathcal{D}(S)$ is the next-action function. If the MDP is in $s \in S_1$ and $m \in M$ is the current memory element, then the strategy chooses $s'$, the next state of the MDP, according to the probability distribution $\alpha_n(m,s)$. When the MDP leaves a state $s \in S$, the memory is updated to $\alpha_u(m, s)$. Hence updates are deterministic and outputs are potentially stochastic. Formally, $(M, m_0, \alpha_u, \alpha_n)$ defines the strategy $\lambda$ such that $\lambda(\rho\cdot s) = \alpha_n(\widehat{\alpha}_u(m_0, \rho), s)$ for all $\rho \in \mathtt{Pref}(\mymdp)$ and $s \in S_1$, where $\widehat{\alpha}_u$ extends $\alpha_u$ to sequences of states starting from $m_0$ as expected. Note that pure strategies have deterministic next-action functions. A strategy $\lambda$ is \textit{finite-memory} if $\vert M\vert < \infty$ and \textit{memoryless} if $\vert M\vert = 1$. That is, it does not depend on the history but only on the current state of the MDP: in this case, we have that $\lambda\colon S_1 \rightarrow \mathcal{D}(S)$. Finally, if the same strategy can be used regardless of the initial state, we say that a \textit{uniform} strategy exists.

A play $\pi$ is said to be \textit{consistent} with a strategy $\lambda$ if for all $n \geq 0$ such that $\pi(n) \in S_1$, we have that $\pi(n+1) \in \mathtt{Supp}(\lambda(\pi(0,n))$. This notion is defined similarly for prefixes. 
We denote by $\out^{\mymdp}(\lambda)\subseteq \mathtt{Plays}(G)$ the set of plays consistent with $\lambda$.
We use $\out^{\mymdp}_s(\lambda)$ when fixing an initial state $s$.

Given a strategy $\lambda$ in $\mymdp$ and a prefix $\rho \in \mathtt{Pref}_1(\mymdp)$, we define the \textit{initialized strategy} $\lambda[\rho]$ as follows:
$$\forall\, \rho'\in\mathtt{Pref}_1(\mymdp),\, \lambda[\rho](\rho')=\begin{cases}
	\lambda(\rho\cdot\rho')\text{ if } \rho\cdot\rho' \in \mathtt{Pref}_1(\mymdp),\\
	\lambda(\rho')\text{ otherwise.}\\
\end{cases} $$

\smallskip\noindent\textbf{Markov chain induced by a strategy.} An MDP $\mymdp = (G = (S, E), S_1, S_2, \delta)$ and a strategy $\lambda$ for $\playerOne$ encoded by the stochastic Moore machine $\mathtt{M} = (M, m_0, \alpha_u, \alpha_n)$ determine an MC $\mc = (G', \delta')$ on the state space  $S' = S \times M$ as follows. For any pair of states $s'_1 = (s_1, m_1)$ and $s'_2 = (s_2, m_2)$ in $S'$, $\delta'(s'_1)(s'_2) = \alpha_n(s_1,m_1)(s_2)$ if $m_2 = \alpha_u(s_1,m_1)$ and $0$ otherwise. Observe that given a finite MDP, a finite-memory (resp. infinite-memory) strategy induces a finite (resp. infinite) MC. We define plays and prefixes of an induced MC as before, considering only transitions with non-zero probability.

Let $\mymdp = (G, S_1, S_2, \delta)$ be an MDP, $s \in S$ an initial state and $\lambda$ a strategy of $\playerOne$ encoded by $\mathtt{M} = (M, m_0, \alpha_u, \alpha_n)$. Let $\mathcal{A} \subseteq \mathtt{Plays}(G)$. We denote by $\prob^{\lambda}_{\mymdp, s}[\mathcal{A}]$ the \textit{probability} of the plays of the induced MC $\mc$ (starting in $(s, m_0)$) whose projection (defined by removing the memory elements) to $\mymdp$ is in $\mathcal{A}$, i.e., the probability of event $\mathcal{A}$ when $\mymdp$ is executed with initial state $s$ and strategy $\lambda$. Note that every Borel set $\mathcal{A}$ has a uniquely defined probability~\cite{vardi1985automatic} (Carathéodory's extension theorem induces a unique probability measure on the Borel $\sigma$-algebra over $\mathtt{Plays}(G)$). Throughout the paper, we may also write $\prob^{\lambda}_{\mymdp, s}[\mathcal{A}]$ for $\mathcal{A}$ a set of \textit{prefixes}: this is an abuse of notation meaning that we consider the probability of the set of plays $\mathcal{A}'$ containing all consistent continuations of prefixes in $\mathcal{A}$ (i.e., we consider the probability of the cylinders defined by those prefixes).

\smallskip\noindent\textbf{Two-player games.} Observe that in MDPs, we call both $\playerOne$ and $\playerTwo$ \textit{players}, though $\playerTwo$ does not have controllable choices. Indeed, an MDP can be seen as a game where $\playerOne$ plays against a stochastic adversary $\playerTwo$ using the fixed memoryless randomized strategy $\delta$: they are sometimes called $1\frac{1}{2}$-player games. Throughout this paper, we also need to consider real two-player games, denoted $\game = (G = (S,E), S_1, S_2)$. In those games, both players control their actions using strategies: we extend all previous notions (e.g., strategies) to $\playerTwo$ and games.
In particular, any MDP $\mymdp = (G, S_1, S_2, \delta)$ can be seen as a two-player game by forgetting about the transition function $\delta$ and letting $\playerTwo$ pick the strategy of its choice. Mixing both interpretations (stochastic and antagonistic versions of $\playerTwo$) is crucial in the beyond worst-case framework that we study here.

\smallskip\noindent\textbf{Objectives.} Given an MDP $\mymdp = (G, S_1, S_1, \delta)$, an \textit{objective} is a set of plays $\mathcal{A} \subseteq \mathtt{Plays}(G)$. We consider two classical objectives from the literature. Both define measurable events. To define them, we introduce the following notation: given a play $\pi \in \mathtt{Plays}(G)$, let
\(
\infny(\pi) = \{ s \in S \mid \forall\, i \geq 0,\, \exists\, j \geq i,\, \pi(j) = s\}
\)
be the set of states seen infinitely often along $\pi$.

\begin{itemize}
\item \textit{Reachability.} Given a target $T \subseteq S$, the reachability objective asks for plays that visit $T$:
\(
\mathtt{Reach}(T) = \{\pi \in \mathtt{Plays}(G) \mid \exists\, n \geq 0,\, \pi(n) \in T\}.
\)
We later use the LTL notation $\Diamond T$ to denote the event $\mathtt{Reach}(T)$.
\item \textit{Parity.} Let $p\colon S \rightarrow \{1, 2, \dotso, d\}$ be a \textit{priority function} that maps each state to an integer priority, where $d \leq \vert S\vert +1$ (w.l.o.g.). The parity objective asks that, among the priorities seen infinitely often, the \textit{maximal} one be even: $\mathtt{Parity}(p) = \{\pi\in\mathtt{Plays}(G) \mid {\displaystyle\max_{s \in \infny(\pi)} p(s)} \text{ is even}\}$. We later simply use $p$ to denote the event $\mathtt{Parity}(p)$.
\end{itemize}

\subsection{Technical tools}

For the following definitions, let $\mymdp = (G= (S, E), S_1, S_2, \delta)$ be an MDP.

\smallskip\noindent\textbf{Attractors.} The \textit{attractor} for $\player{i}$, $i \in \{1, 2\}$, of a target set of states $T$, denoted $\mathtt{Attr}_{i}(T)$ is computed as the fixed point of the sequence 
\(
\mathtt{Attr}^{n+1}_{i}(T) = \mathtt{Attr}^{n}_{i}(T) \cup \{ s \in S_i \mid \mathtt{Succ}(s) \cap \mathtt{Attr}^{n}_{i}(T) \neq \emptyset\} \cup \{ s \in S_{3-i} \mid \mathtt{Succ}(s) \subseteq \mathtt{Attr}^{n}_{i}(T)\}
\)
with $\mathtt{Attr}^{0}_{i}(T) = T$. It contains all states from which $\player{i}$ can force a visit of $T$. Observe that this notion interprets $\playerTwo$ as an antagonistic adversary, i.e., having the choice of his strategy.

\smallskip\noindent\textbf{Traps, end-components and sub-MDPs.}
A \textit{trap} for $\player{i}$, $i \in \{1, 2\}$, is a set $R\subseteq S$ that $\player{i}$ cannot leave: $\forall\, s\in R\cap S_i,\, \mathtt{Succ}(s)\subseteq R$ and $\forall\, s\in R\cap S_{3-i}$, $\mathtt{Succ}(s)\cap R\neq \emptyset$.

An \textit{end-component} (EC) of $\mymdp$ is a trap $C$ for $\playerTwo$ that is \textit{strongly connected}, i.e., for any two states $s, s' \in C$, there exists a path from $s$ to $s'$ that stays in $C$. It is well-known that inside an EC~$C$, $\playerOne$ can force the visit of any state $s \in C$ with probability $1$ (that is, when $\playerTwo$ is seen as stochastic and obeys the strategy $\delta$), see e.g.,~\cite{baier2008principles}. The union of two ECs with non-empty intersection is an EC. An EC $C$ is thus \textit{maximal} if, for every EC $C'$, $C' \subseteq C \vee C' \cap C = \emptyset$.

Given an EC $C \subseteq S$ of $\mymdp$, we denote by $\mymdp_{\downharpoonright C}$ the \textit{sub-MDP} that is naturally defined by $\mymdp_{\downharpoonright C} = (G' = (C, E \cap C \times C), S'_1 = S_1 \cap C, S'_2 = S_2 \cap C, \delta')$, where $\delta'\colon S'_2 \rightarrow \mathcal{D}(C)$ is simply the restriction of $\delta$ to the domain $C$. Note that $\mymdp_{\downharpoonright C}$ is a well-defined MDP: it has no deadlock since $C$ is strongly connected and in all stochastic states $s$, the support of $\delta'(s)$ is included in $C$ (as $C$ was an EC in $\mymdp$).

\smallskip\noindent\textbf{Technical lemmas.} Before turning to the beyond worst-case problem, we recall some classical results about MDPs that will be useful later on.

\begin{restatable}[Optimal reachability~\cite{baier2008principles}]{lemma}{lemma_opti_reach}
\label{lemma_opti_reach}
Given an MDP $\mymdp = (G = (S, E), S_1, S_2, \delta)$ and a target set $T \subseteq S$, we can compute for each state $s \in S$ the maximal probability $v^{\ast}_{s} = \sup_{\lambda \in \Lambda} \prob^{\lambda}_{\mymdp, s}[\Diamond T]$ to reach $T$, in polynomial time. There is an optimal uniform pure memoryless strategy $\lambda^{\ast}$ that enforces $v^{\ast}_{s}$ from all $s \in S$.

Fix $s \in S$ and $c \in \mathbb{Q}$ such that $c < v^{\ast}_{s}$. Then there exists $k \in \mathbb{N}$ such that by playing $\lambda^{\ast}$ from $s$ for $k$ steps, we reach $T$ with probability larger than $c$.
\end{restatable}

We say that $\lambda^{\ast}$ is \textit{optimal} as no strategy can achieve a probability strictly higher than $v^{\ast}_{s}$, for any $s \in S$.

\begin{restatable}[Long-run appearance of ECs~\cite{baier2008principles}]{lemma}{lemma_as_ec}
\label{lemma_as_ec}
Given an MDP $\mymdp = (G = (S, E), S_1, S_2, \delta)$ and $\mathcal{E} = \{C \subseteq S \mid C \text{ is an EC in } \mymdp\}$ the set of all its end-components, for any strategy $\lambda$ of $\playerOne$ and any state $s \in S$, the following holds:
\[
\prob^{\lambda}_{\mymdp, s}\left[ \{\pi \in \out^{\mymdp}(\lambda) \mid \infny(\pi) \in \mathcal{E}\}\right] = 1. 
\]
\end{restatable}

\subsection{The beyond worst-case framework}

\smallskip\noindent\textbf{Events and probabilistic operators.} Consider an MDP $\mymdp = (G = (S, E), S_1, S_2, \delta)$. Recall that we have defined two types of measurable events (specific subsets of $\mathtt{Plays}(G)$) with respective notations $\Diamond T$ for $T \subseteq S$ (reachability), %
and $p$ for $p\colon S \rightarrow \{1, \dotso, d\}$ a priority function (parity).
We define three operators to reason about the probabilities of these events: $\s$, $\p{\sim c}$, and $\as$. Given an event $\mathcal{A}$ and a state $s$, they are used as follows:
\begin{itemize}
\item $\mathcal{A}$ is \textit{sure} from $s$, denoted $s \models \s(\mathcal{A})$, if there exists a strategy $\lambda$ of $\playerOne$ such that $\out^{\mymdp}_{s}(\lambda) \subseteq \mathcal{A}$. Here probabilities are ignored and we consider $\playerTwo$ as antagonistic.
\item $\mathcal{A}$ \textit{holds with probability at least equal to (resp.~greater than)} $c \in \mathbb{Q}$ from $s$, denoted $s \models \p{\geq c}(\mathcal{A})$ (resp.~$s \models \p{> c}(\mathcal{A})$) if there exists $\lambda$ such that $\prob^{\lambda}_{\mymdp, s}[\mathcal{A}] \geq c$ (resp. $> c$).
\item $\mathcal{A}$ is \textit{almost-sure} from $s$, denoted $s \models \as(\mathcal{A})$, if there exists $\lambda$ such that $\prob^{\lambda}_{\mymdp, s}[\mathcal{A}] = 1$.
\end{itemize}
For any operator~$\mathtt{O}$, we say that such a $\lambda$ is a \textit{witness strategy} for $s \models \mathtt{O}(\mathcal{A})$ and we write $s, \lambda \models \mathtt{O}(\mathcal{A})$ to denote it. We will also consider \textit{combinations} of the type $s \models \mathtt{O}_1(\mathcal{A}_1) \wedge \mathtt{O}_2(\mathcal{A}_2)$ for two operators and events: in this case, we require that the same strategy be a witness for both conjuncts, i.e., that there exists $\lambda$ such that $s, \lambda \models \mathtt{O}_1(\mathcal{A}_1)$ and $s, \lambda \models \mathtt{O}_2(\mathcal{A}_2)$. Finally, we will sometimes use different MDPs, in which case we add the considered MDP $\mymdp$ as a subscript on $\models$, e.g., $s \models_{\mymdp} \mathtt{O}(\mathcal{A})$. We drop this subscript when the context is clear.

\smallskip\noindent\textbf{Beyond worst-case problems.} Let $\mymdp = (G = (S, E), S_1, S_2, \delta)$ be an MDP, $s \in S$ be an initial state, and $p_1$, $p_2$ be two priority functions on $\mymdp$. Throughout this paper, we provide algorithms to decide the existence of a witness strategy --- and synthesize it --- for the following formulae combining worst-case and probabilistic guarantees regarding parity objectives:
\begin{enumerate}
\item $s \models \s(p_1) \wedge \as(p_2)$ --- this is the canonical beyond worst-case (BWC) problem;
\item $s \models \s(p_1) \wedge \p{\sim c}(p_2)$ for $\sim\, \in \{>, \geq\}$ and $c \in \mathbb{Q} \cap [0, 1)$.
\end{enumerate}

\section{Reachability under parity constraints}
\label{sec:reach}

In this section, we consider the synthesis of strategies that enforce to reach a target set of states $T \subseteq S$ with high probability while enforcing a parity condition surely. Using the previously defined notations, we consider two variants of the problem, given a state $s \in S$ and a priority function $p\colon S \rightarrow \{1, \dotso d\}$:
\begin{enumerate}
\item $s \models \s(p) \wedge \as(\Diamond T)$,
\item $s \models \s(p) \wedge \p{\sim c}(\Diamond T)$ for $\sim\, \in \{>, \geq\}$ and $c \in \mathbb{Q} \cap [0, 1)$.
\end{enumerate}

\subsection{Almost-sure reachability}

\begin{theorem}
\label{thm:reach-prob-one}
Given an MDP $\mymdp = (G = (S,E), S_1, S_2, \delta)$, a state $s_{0} \in S$, a priority function $p\colon S \rightarrow \{1, \dotso, d\}$, and a target set of states $T \subseteq S$, it can be decided in ${\sf NP} \cap {\sf coNP}$ if  $s_0 \models \s(p) \wedge \as(\Diamond T)$. If the answer is $\yes$, then there exists a finite-memory witness strategy. This decision problem is at least as hard as solving parity games.
\end{theorem}
\begin{proof}
A slight variant of this result has been established inside the proof of~\cite[Lemma 3]{almagor2016minimizing} (extended version of~\cite{DBLP:conf/concur/AlmagorKV16}), where it is asked that the parity condition holds not necessarily on all consistent plays but only on all plays that do not reach the target set $T$. We reduce our problem to the latter as follows. First, we compute all states in $S$ from which $\playerOne$ can enforce $p$ surely. This boils down to solving a classical parity game~\cite{DBLP:journals/ipl/Jurdzinski98}. We remove all states that are not surely winning for $p$ in $\mymdp$ as entering one of those states implies that the parity condition will be violated on some consistent play. Then we apply the approach described in Theorem~\ref{thm:reach-Almagor} (Appendix~\ref{app:density}) to the resulting MDP: it is based on a reduction to a particular parity-B\"uchi game based on the MDP. If the answer is $\yes$ for $s_0$, we answer $\yes$ to our problem, otherwise we answer $\no$. We claim that this reduction is correct. First, it is clear that satisfying their constraint is a necessary condition to satisfy ours. Now, if there is a strategy $\lambda$ that ensures to reach $T$ almost-surely in the reduced MDP and enforces $p$ on all the consistent plays that never reach $T$, then we construct a strategy $\lambda'$ that plays as $\lambda$ up to reaching $T$ (if ever) and then switches to the strategy $\lambda^p$ that enforces $p$ surely. This new strategy is a witness for our problem since the parity objective is prefix-independent. That is, $s_0, \lambda' \models \s(p) \wedge \as(\Diamond T)$. Since both solving the parity game and solving the parity-B\"uchi game from Theorem~\ref{thm:reach-Almagor} are in ${\sf NP} \cap {\sf coNP}$, we conclude that our decision problem lies in ${\sf P}^{{\sf NP} \cap {\sf coNP}} = {\sf NP} \cap {\sf coNP}$~\cite{Bra79}. Furthermore, our problem clearly generalizes parity games: it suffices to fix $T = S$ to trivially ensure the second conjunct and obtain a classical parity game. Hence, the hardness follows.

Finally, since $\lambda^p$ is memoryless w.l.o.g.~and the strategy $\lambda$ obtained via Theorem~\ref{thm:reach-Almagor} is also finite-memory (due to the reduction to a parity-B\"uchi game), we have that our witness strategy $\lambda'$ is also finite-memory.
\end{proof}

\subsection{Reachability with threshold probability}

\smallskip\noindent\textbf{Optimal reachability strategies.} We first study strategies that maximize the probability of reaching a target set of states $T \subseteq S$ in an MDP $\mymdp$. Recall that Lemma~\ref{lemma_opti_reach} states the existence of an optimal uniform pure memoryless strategy $\lambda^{\ast}$ that enforces $v^{\ast}_{s}$ from all $s \in S$, where $v^{\ast}_{s}$ is the maximal probability to reach $T$ that can be achieved by $\playerOne$ in $s$. We define the set $E^{\neg \text{opt}} = \{(s, s') \in E \mid s \in S_1 \wedge v^{\ast}_{s} > v^{\ast}_{s'}\}$ that contains all edges that are \textit{non-optimal choices} for $\playerOne$ in the sense that they result in a strict decrease of the probability to reach $T$.

We show that playing, for a finite number of steps, edges that are optimal for reachability (i.e., edges in $E^{\text{opt}} = E \setminus E^{\neg \text{opt}}$), and then switching to an optimal strategy to reach $T$, like $\lambda^{\ast}$, produces an optimal strategy too.

\begin{lemma}
\label{lem:geqk}
Let $\lambda^{\ast}$ be an optimal uniform pure memoryless strategy in $\mymdp$ to reach $T$, from all states in $S$.
If $\lambda$ is a strategy that plays only edges in $E^{opt}$ for $m$ steps, for $m \in \mathbb{N}$, and then switches to $\lambda^{\ast}$, then $\lambda$ is also optimal to reach $T$ from all states in $S$.
\end{lemma}

\begin{proof}
We prove this lemma by induction. For $m=0$, we only play $\lambda^{\ast}$, thus $\lambda$ is trivially optimal. Now, assume the property holds up to $m-1$ steps, for $m > 0$, i.e., that all strategies playing in $E^{\text{opt}}$ for $m-1$ steps, and then playing $\lambda^{\ast}$, are optimal. Consider a strategy $\lambda^m$ that does this for $m$ steps now. Strategy $\lambda^m$ thus chooses an edge in $E^{opt}$, and then switches to a strategy $\lambda^{m-1}$ that satisfies the induction hypothesis. So for any state $s_0 \in S_1 \setminus T$ (for $s_0 \in T$, we have probability one whatever the strategy), we have that:
\begin{align*}
\prob^{\lambda^m}_{\mymdp, s_0}\left[\Diamond T\right] &= \sum_{s \in S} \left[ \prob^{\lambda^m}_{\mymdp, s_0}\left[\{\pi \mid \pi(1) = s\} \right] \cdot \prob^{\lambda^{m-1}}_{\mymdp, s}\left[\Diamond T\right]\right]\\
&= \sum_{s \in S} \left[ \prob^{\lambda^m}_{\mymdp, s_0}\left[\{\pi \mid \pi(1) = s\} \right] \cdot v^{\ast}_s\right] &\text{(by induction hypothesis)}\\
&\geq \sum_{s \in S} \left[ \prob^{\lambda^m}_{\mymdp, s_0}\left[\{\pi \mid \pi(1) = s\} \right] \cdot v^{\ast}_{s_0}\right] = v^{\ast}_{s_0}. &\text{(by definition of } E^{\text{opt}})
\end{align*}
Hence, we obtain that $\lambda^m$ is also optimal from $s_0$, which concludes our proof.
\end{proof}

\smallskip\noindent\textbf{Solving the problem.} We now turn to the problem $s_0 \models \s(p) \wedge \p{\sim c}(\Diamond T)$.

\begin{theorem}
\label{theorem_reach_p}
Given an MDP $\mymdp = (G = (S,E), S_1, S_2, \delta)$, a state $s_0 \in S$, a priority function $p\colon S \rightarrow \{1, \dotso, d\}$, a target set of states $T \subseteq S$, and a probability threshold $c \in [0, 1) \cap \mathbb{Q}$, it can be decided in ${\sf NP} \cap {\sf coNP}$ if $s_0 \models \s(p) \wedge \p{\sim c}(\Diamond T)$ for $\sim\, \in \{>, \geq\}$. If the answer is $\yes$, then there exists a finite-memory witness strategy. This decision problem is at least as hard as solving parity games.
\end{theorem}

\begin{proof}
First, we remove from $\mymdp$ the states from which $\playerOne$ cannot enforce $\s(p)$. In particular if $s_0 \not\models \s(p)$, the answer is $\no$. Remember that it can be done by computing the winning states in $\mymdp$ seen as a parity game: a problem that belongs to ${\sf NP} \cap {\sf coNP}$~\cite{DBLP:journals/ipl/Jurdzinski98}. We now reason on the MDP $\mymdp^{\text{w}}$ where those losing states have been removed: it is well-defined as the set of winning states for $\playerOne$ in $\mymdp$ is a trap for $\playerTwo$. We note $\lambda^p$ a uniform pure memoryless strategy that ensures $\s(p)$ from any state of $\mymdp^{\text{w}}$.
The proof of this theorem is easy for the case $>$, and more involved for $\geq$. We start with the easy case.

The algorithm for the case $>c$ is as follows. First, we compute the maximal probability $v_{s_0}^{\ast}$ to reach $T$, as in Lemma~\ref{lemma_opti_reach}. This takes polynomial time. If $v_{s_0}^{\ast} \leq c$, then the answer is clearly $\no$. Otherwise, we claim that the answer is $\yes$. Let $\lambda^{\ast}$ be an optimal uniform pure memoryless strategy enforcing probability $v_{s_0}^{\ast}$ from $s_0$. We construct a witness strategy $\lambda$ for $s_0 \models \s(p) \wedge \p{> c}(\Diamond T)$ from $\lambda^{\ast}$ and $\lambda^{p}$ as follows.
First, starting in $s_0$, the strategy $\lambda$ plays as $\lambda^{\ast}$ for $k$ steps where $k$ is taken as in Lemma~\ref{lemma_opti_reach}: hence the probability to reach $T$ after $k$ steps is strictly greater than $c$, which implies that $s_0, \lambda \models \p{> c}(\Diamond T)$. Then, $\lambda$ switches to $\lambda^{p}$. Since the parity objective is prefix-independent, we have that $s_0, \lambda \models \s(p)$, and we are done. Overall, our procedure lies in ${\sf P}^{{\sf NP} \cap {\sf coNP}} = {\sf NP} \cap {\sf coNP}$~\cite{Bra79}. Also observe that our strategy $\lambda$ is finite-memory since $\lambda^\ast$ and $\lambda^p$ are memoryless and $k$ is finite thanks to Lemma~\ref{lemma_opti_reach}.

We now turn to the case $\geq c$. The algorithm, as in the previous case, first computes the maximal probability $v_{s_0}^{\ast}$ to reach $T$, in polynomial time. If $v_{s_0}^{\ast} > c$, then we answer $\yes$ as we can apply the same reasoning as in the previous case. If $v_{s_0}^{\ast} < c$, then we trivially answer $\no$. The more involved case is when $v_{s_0}^{\ast} = c$. In this case, we must verify that probability $c$ is still achievable if, in addition, it is required to enforce $\s(p)$. To answer this question, we modify $\mymdp^{\text{w}}$ and we reduce our problem to the almost-sure case 
of Theorem~\ref{thm:reach-prob-one}. We present the details in the following.
Intuitively, we construct the MDP $\mymdp'$ as follows:
  \begin{itemize}
	\item we enrich states with one bit that records if the set $T$ has been visited or not;
	\item while $T$ has not been visited, we suppress all edges controlled by $\playerOne$ that are not optimal for reachability, i.e., all edges in $E^{\neg \text{opt}}$;
	\item while $T$ has not been visited, we delete all states that cannot reach $T$ (i.e., all states that do not have a path to $T$ in the underlying graph) and normalize the probability of the edges that survive this deletion: we note the set of states to be removed by $S^{\text{w}}_{\Box \neg T}$. 
  \end{itemize}

  Formally, the new MDP $\mymdp'$ is defined as follows:
  \begin{itemize}
  	\item $S'=(S^{\text{w}} \times \{0,1\}) \setminus ( S^{\text{w}}_{\Box \neg T}\times \{0\})$, $S'_1=S' \cap (S^{\text{w}}_1 \times \{0,1\})$, and $S'_2=S' \cap (S^{\text{w}}_2 \times \{0,1\})$;
	\item $s'_0$ is equal to $(s_0,0)$ if $s_0 \not\in T$, and to $(s_0,1)$ if $s_0 \in T$;
	\item $E' = \{ ((s,0),(s',0)) \mid (s,s') \in E^{\text{opt}} \land s' \not\in T\}$ $\cup$ $\{ ((s,0),(s',1)) \mid (s,s') \in E^{\text{opt}} \land s' \in T\}$ $\cup$ $\{ ((s,1),(s',1)) \mid (s,s') \in E \}$;
	\item the probability distributions are then defined as:
		\begin{itemize}
		\item for all $(s,0) \in S'_2$,  and $((s,0),(s',b)) \in E'$, $\delta'((s,0))((s',b)) = \frac{\delta(s)(s')}{X_s}$ where we define $X_s = \displaystyle\sum_{t \in S^{\text{w}} \mid ((s,0),(t,b)) \in E'} \delta(s)(t)$,
		
			\item for all $(s,1) \in S'_2$,  and $((s,1),(s',1)) \in E'$, $\delta'((s,1))((s',1))=\delta(s)(s')$;
		\end{itemize}
	\item the priority function $p'\colon S' \rightarrow \{1,\dotso,d\}$ is such that for all $(s,b) \in S'$, $p'((s,b))=p(s)$.
  \end{itemize}

Note that each prefix (resp.~play) $\rho'$ in $\mymdp'$ is naturally mapped to a unique prefix $\gamma(\rho')$ in $\mymdp^{\text{w}}$, obtained by projecting away the additional bit of information (about visiting the target~$T$). 
Conversely, each prefix (resp.~play) $\rho$ in $\mymdp^{\text{w}}$ that does not use edges in $E^{\neg {\sf opt}}$ and states in $S^{\text{w}}_{\Box \neg T}$ is naturally mapped to a unique prefix $\alpha(\rho)$ in $\mymdp'$. Those mappings are such that for all $\rho'$ in $\mymdp'$, $\rho'=\alpha(\gamma(\rho'))$.
  
Finally, we claim that $s_0 \models \s(p) \wedge \p{\geq c}(\Diamond T)$ in $\mymdp^{\text{w}}$ if, and only if, $s'_0 \models \s(p) \wedge \as(\Diamond T')$ in $\mymdp'$, where $T' = T \times \{1\}$. For the sake of readability, we postpone the correctness of this reduction to Lemma~\ref{lemma_geq_to_as}. Given that the latter problem was shown to be in ${\sf NP} \cap {\sf coNP}$ in Theorem~\ref{thm:reach-prob-one}, and that $\mymdp'$ is only polynomially larger than $\mymdp^{\text{w}}$ (and thus $\mymdp$), we obtain the claimed complexity for the problem, using ${\sf P}^{{\sf NP} \cap {\sf coNP}} = {\sf NP} \cap {\sf coNP}$~\cite{Bra79}. The proof of Lemma~\ref{lemma_geq_to_as} also implies that the witness strategy can be finite-memory.

Again, problem $s_0 \models \s(p) \wedge \p{\sim c}(\Diamond T)$ trivially generalizes parity games, by taking $T = S$, which concludes our proof.
\end{proof}

\begin{lemma}[Correctness of the reduction used for Theorem~\ref{theorem_reach_p}]
\label{lemma_geq_to_as}
The following equivalence holds:
\[
s'_0 \models_{\mymdp'} \s(p) \wedge \as(\Diamond T') \iff s_0 \models_{\mymdp^{\textsc{w}}} \s(p) \wedge \p{\geq c}(\Diamond T).
\]
\end{lemma}

\begin{proof}
We start with the \textit{left-to-right} implication. Assume that there exists a witness strategy $\lambda'$ for $\playerOne$ in $\mymdp'$ such that $s'_0,\lambda' \models_{\mymdp'} \s(p) \wedge \as(\Diamond T')$. By Theorem~\ref{thm:reach-prob-one}, we may assume $\lambda'$ to be finite-memory. From $\lambda'$, we build a strategy $\lambda$ in $\mymdp^{\text{w}}$ as follows:
  \begin{itemize}
	\item for all $\rho \in \mathtt{Pref}(G)$ that does not visit $T$, only uses edges in $E^{\text{opt}}$ and does not contain states in $S^{\text{w}}_{\Box \neg T}$, $\lambda(\rho) = \gamma(\lambda'(\alpha(\rho)))$, i.e., $\lambda$ mimics $\lambda'$,
	\item for all other $\rho  \in \mathtt{Pref}(G)$, the strategy plays as $\lambda^p$, the memoryless strategy ensuring $\s(p)$ in $\mymdp^{\text{w}}$.
  \end{itemize}
Observe that $\lambda$ is a finite-memory strategy: the conditions to check can be encoded using a finite number of memory elements in addition to the finite Moore machine representing $\lambda'$. We now claim that $s_0, \lambda \models_{\mymdp^{\textsc{w}}} \s(p) \wedge \p{\geq c}(\Diamond T)$.

First, consider the conjunct $\s(p)$. Along any play $\pi \in \mathtt{Plays}(G)$ starting in $s_0$, either (i) $\lambda$ always plays as $\lambda'$, or (ii) it eventually switches to $\lambda^p$. In case (i), $\pi$ has a unique corresponding play in $\mymdp'$, $\pi' = \alpha(\pi)$, and $\pi'$ is consistent with $\lambda'$. Since $\lambda'$ ensures $\s(p)$ by hypothesis and the priority function is the same in $\mymdp'$ and $\mymdp^{\text{w}}$, $\pi$ also satisfies the parity objective. In case~(ii), $\pi$ also satisfies the parity objective by definition of $\lambda^p$ and prefix-independence of parity. Hence, we have that $s_0, \lambda \models_{\mymdp^{\textsc{w}}} \s(p)$.

It remains to prove that $s_0, \lambda \models_{\mymdp^{\textsc{w}}} \p{\geq c}(\Diamond T)$. We prove it by contradiction. Assume that $\prob^{\lambda}_{\mymdp^{\text{w}}, s_0}\left[\Diamond T\right] < c$, or equivalently, that $\prob^{\lambda}_{\mymdp^{\text{w}}, s_0}\left[\square \neg T\right] > 1 - c$, using the classical \textit{always} LTL operator. This means that, because of the strict inequality $> 1 - c$, there exists $m$ in $\mathbb{N}$ such that the set
\begin{align*}
\mathtt{BadPref} = \Big\lbrace \rho \in \mathtt{Pref}(G) &\mid \rho \text{ consistent with } \lambda \wedge \vert \rho \vert = m + 1\\
&\wedge \forall\, \text{ continuation } \pi \text{ of } \rho \text{ consistent with } \lambda,\, \pi \not\models \Diamond T\Big\rbrace
\end{align*}
satisfies $ \sum_{\rho \in \mathtt{BadPref}} \prob^{\lambda}_{\mymdp^{\text{w}}, s_0}\left[\rho\right] > 1 - c$. Note that this sum is indeed a lower bound on the probability of $\square \neg T$ as all prefixes of $\mathtt{BadPref}$ define disjoint cylinder sets.

Now, we claim that for all $\rho \in \mathtt{BadPref}$, we have that $\mathtt{Last}(\rho) \in S^{\text{w}}_{\Box \neg T}$. To prove it, recall the definition of $\lambda$. If it were not true, then $\lambda$ would have mimicked $\lambda'$ all along $\rho$ (i.e., not switching to $\lambda^p$). But since all continuations of $\rho$ consistent with $\lambda$ in $\mymdp^{\text{w}}$ never reach $T$, this would also hold for continuations of $\alpha(\rho)$ consistent with $\lambda'$ in $\mymdp'$. Hence prefix $\alpha(\rho)$ would prove that $\prob^{\lambda'}_{\mymdp', s'_0}\left[\Diamond T'\right] < 1$, which contradicts the hypothesis stating that $s'_0, \lambda' \models \as(\Diamond T')$. Hence, we can now confirm that for all $\rho \in \mathtt{BadPref}$, we have that $\mathtt{Last}(\rho) \in S^{\text{w}}_{\Box \neg T}$.

To pursue, we define a new strategy on $\mymdp^{\text{w}}$, $\lambda^r$: it plays as $\lambda$, which only using edges in $E^{\text{opt}}$, for $m$ steps ($m$ as defined above), then switches to $\lambda^\ast$, the optimal strategy to reach $T$ in $\mymdp^{\text{w}}$, which ensures probability $v_{s_0}^{\ast} = c$ of reaching $T$ from $s_0$ (see Proof of Theorem~\ref{theorem_reach_p}). By Lemma~\ref{lem:geqk}, we have that $\lambda^r$ is still optimal, hence that $\prob^{\lambda^r}_{\mymdp^{\text{w}}, s_0}\left[\Diamond T\right] = c$. But by the previous arguments, we also have that
\begin{align*}
\prob^{\lambda^r}_{\mymdp^{\text{w}}, s_0}\left[\square \neg T\right] \geq \sum_{\rho \in \mathtt{BadPref}} \prob^{\lambda}_{\mymdp^{\text{w}}, s_0}\left[\rho\right] > 1 - c,
\end{align*}
hence that $\prob^{\lambda^r}_{\mymdp^{\text{w}}, s_0}\left[\Diamond T\right] < c$. We have obtained the contradiction, hence proving that $s_0, \lambda \models_{\mymdp^{\textsc{w}}} \p{\geq c}(\Diamond T)$ and concluding the proof of the left-to-right implication.

We now prove the \textit{right-to-left} implication.  Assume that there exists a witness strategy $\lambda$ for $\playerOne$ in $\mymdp^{\text{w}}$ such that $s_0,\lambda \models_{\mymdp^{\textsc{w}}} \s(p) \wedge  \p{\geq c}(\Diamond T)$. Recall that the maximal achievable probability to reach $T$ in $\mymdp^{\textsc{w}}$ is $v_{s_0}^{\ast} = c$. Hence, it can be shown that $\lambda$ must restrict its choices to $E^{\text{opt}}$ before reaching $T$, otherwise it would not be possible to ensure $ \p{\geq c}(\Diamond T)$. Similarly, we can assume w.l.o.g. that $\lambda$ switches to $\lambda^p$, the strategy ensuring $\s(p)$, after $T$ or $S^{\text{w}}_{\Box \neg T}$ is reached: if not the case, we can build a strategy with those properties based on $\lambda$ that is still a witness for the considered property. We first study $\lambda'$, which is simply $\lambda$ played on $\mymdp'$, and show that it is a witness for $s'_0 \models_{\mymdp'} \s(p)$.

To prove that $s'_0, \lambda' \models_{\mymdp'} \s(p)$, observe that on any play $\pi'$ consistent with $\lambda'$, we either (i) reach $T'$ and switch to $\lambda^p$ (more precisely, its translation over $\mymdp'$), (ii) do not reach $T'$ but neither reach $S^{\text{w}}_{\Box \neg T}$ (by construction), hence keep playing like $\lambda$. In case (i), and since all edges from the initial MDP $\mymdp^{\text{w}}$ are present after reaching $T'$, $\lambda^p$ can be implemented faithfully and any consistent play $\pi'$ will satisfy the parity objective thanks to prefix-independence. In case (ii), any consistent play $\pi'$ of $\lambda'$ has a unique corresponding play $\pi = \gamma(\pi')$ in $\mymdp^{\text{w}}$ that is consistent with $\lambda$, hence satisfies the parity objective since $\lambda$ ensures $\s(p)$. We thus conclude that $\lambda'$ is a witness for $s'_0 \models_{\mymdp'} \s(p)$.

We now show that $T'$ is dense in $\lambda'$, i.e., that any prefix $\rho'$ consistent with $\lambda'$ has a continuation $\pi'$ such that $\pi' \models \Diamond T'$. In other words, seeing strategy $\lambda'$ as an execution tree, there is no subtree such that all branches in this subtree never visit~$T'$. Formally, this is saying that the following set must be empty:
\begin{align*}
\mathtt{BadPref}' = \Big\lbrace \rho' \in \mathtt{Pref}(G') &\mid \rho' \text{ consistent with } \lambda'\\
&\wedge \forall\, \text{ continuation } \pi' \text{ of } \rho' \text{ consistent with } \lambda',\, \pi' \not\models \Diamond T'\Big\rbrace.
\end{align*}
We prove its emptiness by contradiction. Assume there exists a prefix $\rho'$ in $\mathtt{BadPref}'$. Since it is a prefix in $\mymdp'$, it implies by construction that $\mathtt{Last}(\gamma(\rho')) \not\in S^{\text{w}}_{\Box \neg T}$ in $\mymdp^{\text{w}}$, hence that some paths reach $T'$ from $\mathtt{Last}(\gamma(\rho'))$ in the underlying graph. Hence, we can define a new strategy $\lambda^r$ in $\mymdp^{\text{w}}$ that plays exactly as $\lambda$ except after prefix $\gamma(\rho')$, where it switches to the optimal reachability strategy $\lambda^{\ast}$: the probability to reach $T$ after $\gamma(\rho')$ under $\lambda^r$ is thus strictly positive, say equal to $q > 0$. We know have that
\[
\prob^{\lambda^r}_{\mymdp^{\text{w}}, s_0}\left[\Diamond T\right] = c + \prob^{\lambda^r}_{\mymdp^{\text{w}}, s_0}\left[\gamma(\rho')\right] \cdot q > c
\]
which is in contradiction with $c = v_{s_0}^{\ast}$ being the maximal achievable probability. Hence, we deduce that $\mathtt{BadPref}'$ is indeed empty and that $T'$ is dense for $\lambda'$. Now, using the reduction to a parity-B\"uchi game which underlies Theorem~\ref{thm:reach-prob-one}, we deduce from the density argument presented in Appendix~\ref{app:density} that we can build $\lambda''$ from $\lambda'$ such that $s'_0, \lambda'' \models_{\mymdp'} \s(p) \wedge \as(\Diamond T')$. This concludes the proof of the right-to-left implication and of this lemma.
\end{proof}

\section{Almost-sure parity under parity constraints}
\label{sec:asp}

\smallskip\noindent\textbf{Overview and key lemma.} Here, we study strategies that enforce a parity objective~$p_1$ surely (i.e., on all consistent plays) and a parity objective $p_2$ almost-surely (i.e., with probability one). Formally, we look at the problem $s \models \s(p_1) \wedge \as(p_2)$. The cornerstone of our approach is the notion of \textit{ultra-good end-component}. Throughout this section, we consider an MDP $\mymdp = (G = (S, E), S_1, S_1, \delta)$ with two priority functions $p_1$ and~$p_2$.

\begin{definition}
\label{def:ugec}
An end-component $C$ of $\mymdp$ is \emph{ultra-good} ($\ugec$) if in the sub-MDP $\mymdp_{\downharpoonright C}$, the following two properties hold:
\begin{itemize}
	\item $\mathbf{(1_U)}$ $\forall\, s\in C,\, s\models_{\mymdp_{\downharpoonright C}} \s(p_1) \wedge \as(\Diamond C^{\max}_{\even}(p_1))$, where
	\[
	C^{\max}_{\even}(p_i) = \big\lbrace s\in C \mid (p_i(s) \text{ is even}) \wedge (\forall\, s'\in C,\, p_i(s') \text{ is odd } \implies p_i(s') < p_i(s))\big\rbrace
	\]
	contains the states with even priorities that are larger than any odd priority in $C$ (this set can be empty for arbitrary ECs but needs to be non-empty for UGECs);
	\item $\mathbf{(2_U)}$ $\forall\, s\in C,\, s\models_{\mymdp_{\downharpoonright C}} \as(p_1) \wedge \as (p_2)$, or equivalently, $s\models_{\mymdp_{\downharpoonright C}} \as(p_1 \cap p_2)$.
\end{itemize}
We introduce the following notations:
  \begin{itemize}
  	\item $\ugec(\mymdp)$ is the set of all $\ugec$s of $\mymdp$, 
	\item $\mathcal{U} = {\displaystyle \cup_{U \in \ugec(\mymdp)} U}$ is the set of states that belong to a $\ugec$ in $\mymdp$.
  \end{itemize}
\end{definition}

Intuitively, within a $\ugec$, $\playerOne$ has a strategy to almost-surely visit $C^{\max}_{\even}(p_1)$ while guaranteeing $\s(p_1)$, and he also has a (generally different) strategy that almost-surely ensures both parity objectives. Figure~\ref{fig:full} gives an example of UGEC: here, the strategy ensuring $\mathbf{(1_U)}$ is to go to $d$, and the strategy ensuring $\mathbf{(2_U)}$ is to go to $b$.

\begin{figure}[tbh]
\centering
\scalebox{0.8}{\begin{tikzpicture}[every node/.style={font=\small,inner sep=1pt}]
\draw (0,0) node[rond,bleu] (s2) {$1,0$};
\draw (-2,0) node[carre,bleu] (s1) {$1,0$};
\draw (-4,0) node[carre,vert] (s0) {$2,0$};
\draw (2,0) node[carre,jaune] (s3) {$4,1$};
\draw (4,0) node[carre,rouge] (s4) {$3,1$};
\draw (-4,0.6) node (l0) {$a$};
\draw (-2,0.6) node (l1) {$b$};
\draw (-0,0.6) node (l2) {$c$};
\draw (2,0.6) node (l3) {$d$};
\draw (4,0.6) node (l4) {$e$};
\draw[-latex] (s1) to (s0);
\draw[-latex] (s2) to (s1);
\draw[-latex] (s3) to (s4);
\draw[-latex] (s2) to (s3);
\draw[-latex] (s0) to[out=330,in=210] (s2);
\draw[-latex] (s4) to[out=210,in=330] (s2);
\draw[-latex] (s1) to[out=45,in=135] (s2);
\end{tikzpicture}}
\caption{This MDP is a UGEC: going to $d$ satisfies $\mathbf{(1_U)}$, whereas going to $b$ satisfies $\mathbf{(2_U)}$.}
\label{fig:full}
\end{figure}
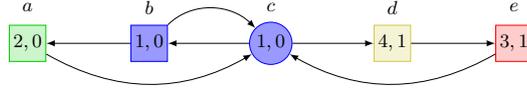

The notion of $\ugec$ strengthens the concept of \textit{super-good} EC that was introduced in~\cite{DBLP:conf/concur/AlmagorKV16}: essentially, the super-good ECs are exactly the ECs satisfying $\mathbf{(1_U)}$. Thus, every $\ugec$ is a super-good EC, but the converse is false.

The central lemma underpinning our approach is the following.
\begin{lemma}
\label{lemma_asp_to_asr}
The following equivalence holds:
\[
s_0 \models \s(p_1) \wedge \as(\Diamond \mathcal{U}) \iff s_0 \models \s(p_1) \wedge \as(p_2).
\]
\end{lemma}
Essentially, this lemma permits to reduce the problem under study to the one treated in Theorem~\ref{thm:reach-prob-one}, provided that we are able to compute $\mathcal{U}$, the set of states appearing in a UGEC. The rest of this section is dedicated to the proof of this lemma and its consequences: in Sect.~\ref{subsec:sufficient}, we prove the left-to-right implication, in Sect.~\ref{subsec:necessary}, we prove its converse, and finally, in Sect.~\ref{subsec:algo}, we establish an algorithm to solve problem $s_0 \models \s(p_1) \wedge \as(p_2)$ based on this reduction and we study its complexity.

\subsection{Left-to-right implication (sufficient condition)}
\label{subsec:sufficient}

\smallskip\noindent\textbf{Available strategies in UGECs}. We first focus on witness strategies for conditions $\mathbf{(1_U)}$ and $\mathbf{(2_U)}$ of Definition~\ref{def:ugec}. 
For $\mathbf{(1_U)}$, it was shown in the proof of~\cite[Lemma 3]{almagor2016minimizing} (extended version of~\cite{DBLP:conf/concur/AlmagorKV16}) that deciding if the condition holds is in ${\sf NP} \cap {\sf coNP}$ and that uniform finite-memory witness strategies exist. For $\mathbf{(2_U)}$, we establish the following lemma.

\begin{lemma}
\label{lemma_strat2}
Let $C$ be an EC of $\mymdp$. The following assertions hold.
\begin{enumerate}
\item It can be decided in polynomial time if condition $\mathbf{(2_U)}$ holds.
\item If the answer is $\yes$, then there exists a (uniform randomized) memoryless witness strategy $\lambda_{2, C}$ and a sub-EC $D \subseteq C$ such that $D^{\max}_{\even}(p_1) \neq \emptyset$, $D^{\max}_{\even}(p_2) \neq \emptyset$, and for all $s \in C$,
\[
\prob^{\lambda_{2, C}}_{\mymdp_{\downharpoonright C}, s}\left[ \{\pi \in \out^{\mymdp_{\downharpoonright C}}(\lambda_{2, C}) \mid \infny(\pi) = D\}\right] = 1.
\]
\item Furthermore, $\lambda_{2, C}$ satisfies the following property: $\forall\, s \in C$, $\forall\, \varepsilon > 0$, $\exists\, n \in \mathbb{N}$ such that
\(
\prob^{\lambda_{2, C}}_{\mymdp_{\downharpoonright C}, s}\left[\left\lbrace \pi \in \out^{\mymdp_{\downharpoonright C}}(\lambda_{2, C}) \mid \exists\, i,\, 0 \leq i \leq  n,\, \pi(i) \in D^{\max}_{\even}(p_1)\right\rbrace\right] \geq 1 - \varepsilon. 
\)
\end{enumerate}
\end{lemma}

\begin{proof}
We first prove Point \textbf{2}. Let $C$ be an EC that satisfies condition $\mathbf{(2_U)}$, and a corresponding witness strategy $\lambda$ (that could use memory in full generality). By Lemma~\ref{lemma_as_ec}, we know that for any state $s \in C$,
\[
\prob^{\lambda}_{\mymdp_{\downharpoonright C}, s}\left[ \{\pi \in \out^{\mymdp_{\downharpoonright C}}(\lambda) \mid \infny(\pi) \in \mathcal{E}\}\right] = 1,
\]
where $\mathcal{E}$ is the set of ECs of $\mymdp_{\downharpoonright C}$. We claim that for any $D \in \mathcal{E}$ for which
\[
\prob^{\lambda}_{\mymdp_{\downharpoonright C}, s}\left[ \{\pi \in \out^{\mymdp_{\downharpoonright C}}(\lambda) \mid \infny(\pi) = D\}\right] > 0,
\]
we necessarily have that $D^{\max}_{\even}(p_1) \neq \emptyset$ and $D^{\max}_{\even}(p_2) \neq \emptyset$. Indeed, assume this is false for $p_i$. Then, with probability strictly greater than zero, $\lambda$ induces plays $\pi$ such that $\max_{s' \in \infny(\pi)} p_i(s')$ is odd (as the maximal priority in $\infny(\pi) = D$ is odd). This contradicts the fact that $\lambda$ is a witness strategy for $s\models_{\mymdp_{\downharpoonright C}} \as(p_1) \wedge \as (p_2)$ (condition $\mathbf{(2_U)}$).

Assume that $\lambda$ is a witness strategy, potentially leading to several ECs $D_j$ satisfying the above property with non-zero probability, and potentially using memory. Since each $D_j$ is included in the EC $C$, it is easy to conceive a witness strategy $\lambda_{2, C}$ that reaches and stays in only one of those ECs with probability one (we call it $D$ from now on) and without using any memory: indeed, each state of an EC can be almost-surely visited using a uniform randomized memoryless strategy~\cite{baier2008principles}.

Second, to prove Point \textbf{3}, it suffices to use Lemma~\ref{lemma_opti_reach}. See that $\lambda_{2, C}$ reaches $D^{\max}_{\even}(p_1)$ almost-surely from any state $s \in C$. By taking $n = k$ for the $k$ given in Lemma~\ref{lemma_opti_reach} for $c = 1 - \varepsilon$, we obtain the desired property.

Finally, consider Point \textbf{1}. We claim that (i) the existence of a sub-EC $D$ such that $D^{\max}_{\even}(p_1) \neq \emptyset$ and $D^{\max}_{\even}(p_2) \neq \emptyset$ is not only necessary but also sufficient to satisfy condition $\mathbf{(2_U)}$, and (ii) the existence of such a set can be decided in polynomial time.

For (i), as discussed above, for any sub-EC of $C$, in particular for $D$, we can build a uniform randomized memoryless strategy $\lambda$ such that 
\[
\prob^{\lambda}_{\mymdp_{\downharpoonright C}, s}\left[ \{\pi \in \out^{\mymdp_{\downharpoonright C}}(\lambda) \mid \infny(\pi) = D\}\right] = 1.
\]
By hypothesis on $D$, we thus have that $p_1$ and $p_2$ are almost-surely satisfied by $\lambda$, hence that $\lambda$ is a witness for $s\models_{\mymdp_{\downharpoonright C}} \as(p_1) \wedge \as (p_2)$ and that condition $\mathbf{(2_U)}$ holds in $C$.

It remains to check the existence of such a set $D \subseteq C$ in polynomial time. First, we check if $C^{\max}_{\even}(p_1) \neq \emptyset$ and $C^{\max}_{\even}(p_2) \neq \emptyset$. If this holds, then $D = C$ and the answer is $\yes$ (it takes linear time obviously). If it does not hold, then we compute the sets
\[
	C^{\max}_{\odd}(p_i) = \big\lbrace s\in C \mid (p_i(s) \text{ is odd}) \wedge (\forall\, s'\in C,\, p_i(s') \text{ is even } \implies p_i(s') < p_i(s))\big\rbrace
	\]
and we iterate this procedure in the sub-EC $C' \subset C$ defined as
\[
C' = C \setminus \mathtt{Attr}_{2}\big(C^{\max}_{\odd}(p_1) \cup C^{\max}_{\odd}(p_2)\big).
\]
It is easy to see that a suitable $D$ exists if and only if this procedure stops before $C' = \emptyset$. In addition, this procedure takes at most $\vert C \vert$ iterations (as we remove at least one state at each step) and each iteration takes linear time. This implies Point \textbf{3} and concludes our proof.
\end{proof}

\smallskip\noindent\textbf{Inside UGECs.} We will now prove that inside any UGEC, there is a strategy for the beyond worst-case problem $\s(p_1) \wedge \as(p_2)$. From now on, let $C$ be a UGEC of $\mymdp$, $\lambda_{1, C}$ be a uniform finite-memory witness strategy for condition $\mathbf{(1_U)}$ of Definition~\ref{def:ugec}, and $\lambda_{2, C}$ be a uniform randomized memoryless one for condition $\mathbf{(2_U)}$, additionally satisfying the properties of Lemma~\ref{lemma_strat2}. 
We build a strategy $\lambda_C$ based on $\lambda_{1, C}$ and $\lambda_{2, C}$.

\begin{definition}
\label{def:strat_ugec}
Let $C \in \ugec(\mymdp)$. Let $(n_i)_{i \in \mathbb{N}}$ be a sequence of naturals $n_i$ such that 
\[
\prob^{\lambda_{2, C}}_{\mymdp_{\downharpoonright C}, s}\left[\left\lbrace \pi \in \out^{\mymdp_{\downharpoonright C}}(\lambda_{2, C}) \mid \exists\, i,\, 0 \leq i \leq  n_i,\, \pi(i) \in D^{\max}_{\even}(p_1)\right\rbrace\right] \geq 1 - 2^{-i}, 
\]
whose existence is guaranteed by Lemma~\ref{lemma_strat2}. We build the beyond worst-case strategy $\lambda_C$ as follows, starting with $i = 0$.
\begin{enumerate}[label=\textbf{\alph*)}]
\item Play $\lambda_{2, C}$ for $n_i$ steps. Then $i = i + 1$ and go to \textbf{b)}.
\item If $D^{\max}_{\even}(p_1)$ was visited in phase \textbf{a)}, then go to \textbf{a)}.

Else, play $\lambda_{1, C}$ until $C^{\max}_{\even}(p_1)$ is reached and then go to \textbf{a)}.
\end{enumerate}
\end{definition}

Observe that $\lambda_C$ requires infinite memory. In the next lemma, we prove that $\lambda_C$ is a proper witness for $\s(p_1) \wedge \as(p_2)$ in the UGEC $C$.

\begin{lemma}
\label{lem:suf}
Let $C \in \ugec(\mymdp)$. For all $s \in C$, it holds that $s, \lambda_C \models \s(p_1) \wedge \as(p_2)$.
\end{lemma}

\begin{proof}
Consider the strategy $\lambda_C$ as described in Definition~\ref{def:strat_ugec} and any state $s \in C$. We first prove that $s, \lambda_C \models \s(p_1)$. We want to show that for any $\pi \in \out^{\mymdp_{\downharpoonright C}}_{s}(\lambda_C)$, $\max_{s' \in \infny(\pi)} p_1(s')$ is even. Fix such a play $\pi$. Three cases are possible: (i) $\lambda_C$ switches infinitely often between $\lambda_{1, C}$ and $\lambda_{2, C}$, (ii) it eventually plays $\lambda_{1, C}$ forever, and (iii) it eventually plays $\lambda_{2, C}$ forever. In case (i), we have that $C^{\max}_{\even}(p_1)$ is visited infinitely often along $\pi$. Since any state in this set has an even priority higher than any odd priority in $C$ by definition, we have that $\pi$ satisfies the parity objective $p_1$. In case (ii), we eventually always play as $\lambda_{1, C}$ and by hypothesis, we know that $s, \lambda_{1, C} \models \s(p_1)$. By prefix-independence, we thus have that $\pi$ satisfies the parity objective $p_1$. Finally, in case (iii), we have that $D^{\max}_{\even}(p_1)$ is visited infinitely often, and that eventually, play $\pi$ never leaves the sub-EC $D$ associated to strategy $\lambda_{2, C}$ (by definition of $\lambda_{2, C}$, the fact that $D^{\max}_{\even}(p_1) \subseteq D$ is reached, and that we never switch to $\lambda_{1, C}$ from some point on). Since any state in $D^{\max}_{\even}(p_1)$ has an even priority higher than any odd priority in $D$, we have that $\pi$ satisfies the parity objective $p_1$. Therefore, we conclude that $s, \lambda_C \models \s(p_1)$.

It remains to show that $s, \lambda_C \models \as(p_2)$. To do so, we will prove that $\lambda_C$ almost-surely ends up in playing only $\lambda_{2,C}$ (which ensures $\as(p_2)$).
First, observe that $\lambda_{1, C}$ almost-surely reaches $C^{\max}_{\even}(p_1)$, hence each phase \textbf{b)} of strategy $\lambda_C$ terminates with probability one. Now, consider the probability $x_i$ that $\lambda_{1, C}$ is never played again after round $i$ (in which $\lambda_{2, C}$ is played for $n_i$ steps). For any round $i \geq 1$, it is equal to $x_i = \prod_{j=i}^{\infty} (1 - 2^{-j})$ thanks to the choice of $n_i$. Since $(1 - 2^{-j}) > 0$ for any $j > 1$, we have that $\log(x_i) = \sum_{j=i}^{\infty} \log(1-2^{-j})$. Now, for any $z \in [0, \frac{1}{2}]$, we have the following inequalities:
\begin{align*}
0 \geq \log(1-z) = - (\log(1) - \log(1-z)) = - \int_{1-z}^{1} \frac{1}{y}\, \mathrm dy \geq -2\int_{1-z}^{1}\, \mathrm dy = -2z.
\end{align*}
Using these, we can bound $\log(x_i)$ as follows:
\begin{align*}
\sum_{j=i}^\infty -2\cdot 2^{-j} \leq \log(x_i) \leq 0 &\iff -2\cdot 2^{-i}\sum_{j=0}^\infty 2^{-j} \leq \log(x_i) \leq 0\\ &\iff -2^{-(i-2)} \leq \log(x_i) \leq 0.
\end{align*}
Hence, $\log(x_i) \xrightarrow{i \to +\infty} 0$, and by continuity of the exponential function, we finally obtain that $x_i \xrightarrow{i \to +\infty} 1$. This proves that $\lambda_C$ almost-surely consists in playing as $\lambda_{2, C}$ forever from some point on, which combined with the prefix-independence of the parity objective, implies that $\lambda_C$ ensures $\as(p_2)$, and concludes our proof.
\end{proof}

\smallskip\noindent\textbf{Global strategy.} We can now prove the left-to-right implication of Lemma~\ref{lemma_asp_to_asr}. For this, assume that for $s_0 \in S$, we have that $\lambda_{\mathcal{U}}$ is a witness for $s_0 \models \s(p_1) \wedge \as(\Diamond \mathcal{U})$, where we recall that $\mathcal{U}$ represents the union of all UGECs of the MDP $\mymdp$. Note that such a strategy can be finite-memory w.l.o.g.~as proved in Theorem~\ref{thm:reach-prob-one}. We build a global strategy $\lambda$ as follows.

\begin{definition}
Based on strategies $\lambda_{\mathcal{U}}$ and $\lambda_{C}$ for all $C \in \ugec(\mymdp)$, we build the global strategy $\lambda$ as follows.
\begin{enumerate}[label=\textbf{\alph*)}]
\item Play $\lambda_{\mathcal{U}}$ until a UGEC $C$ is reached, then go to \textbf{b)}.
\item Play $\lambda_{C}$ forever.
\end{enumerate}
\end{definition} 
This strategy requires infinite memory because it is needed for the strategies $\lambda_C$. We prove that $\lambda$ is a witness for the beyond worst-case problem $\s(p_1) \wedge \as(p_2)$.

\begin{lemma}
\label{lemma_leftright}
It holds that $s_0, \lambda \models \s(p_1) \wedge \as(p_2)$.
\end{lemma}

\begin{proof}
First, we consider $\s(p_1)$. Along any play $\pi$ consistent with $\lambda$, either (i) a UGEC $C$ is eventually reached and $\lambda$ switches to $\lambda_C$, or (ii) $\lambda$ behaves as $\lambda_{\mathcal{U}}$ forever. Since all strategies $\lambda_C$ and strategy $\lambda_{\mathcal{U}}$ ensure $\s(p_1)$ and the parity objective is prefix-independent, we have that $s_0, \lambda \models \s(p_1)$.

Second, with probability one, $\lambda_{\mathcal{U}}$ reaches some UGEC $C$, in which $\lambda_C$ ensures $\as(p_2)$. Again invoking prefix-independence, we conclude that $\lambda$ is also a witness for $s_0 \models \as(p_2)$, which ends our proof.
\end{proof}

\subsection{Right-to-left implication (necessary condition)}
\label{subsec:necessary}

We now turn to the converse implication of Lemma~\ref{lemma_asp_to_asr}, i.e., that $s_0 \models \s(p_1) \wedge \as(p_2)$ implies $s_0 \models \s(p_1) \wedge \as(\Diamond \mathcal{U})$. This in particular requires $\mathcal{U}$ to be non-empty, i.e., the existence of a UGEC in $\mymdp$.

\smallskip\noindent\textbf{Initialized strategies.} We start by an intermediate lemma regarding witness strategies for the beyond worst-case problem. It establishes that all states reachable via such a strategy also satisfy the property.

\begin{lemma}
\label{lem:suf-bwc}
For every state $s \in S$, every strategy $\lambda$ such that $s, \lambda \models \s(p_1) \wedge \as(p_2)$, and every prefix $\rho \in \mathtt{Pref}(\out_s^{\mymdp}(\lambda))$, we have that $\mathtt{Last}(\rho) \models \s(p_1) \wedge \as(p_2)$.
\end{lemma}

\begin{proof}
Fix $s \in S$, $\lambda$ a witness for $s \models \s(p_1) \wedge \as(p_2)$, and a prefix $\rho \in \mathtt{Pref}(\out_s^{\mymdp}(\lambda))$. We claim that the initialized strategy $\lambda[\rho]$, that behaves as $\lambda$ behaves after prefix~$\rho$, is a witness for $\mathtt{Last}(\rho) \models \s(p_1) \wedge \as(p_2)$.

We first consider $\s(p_1)$. For all $\pi \in \out_{\mathtt{Last}(\rho)}^{\mymdp}(\lambda[\rho])$, we have that $\rho \cdot \pi \in \out_s^{\mymdp}(\lambda)$. Furthermore, $\max_{s' \in \infny(\rho \cdot \pi)} p_1(s') = \max_{s' \in \infny(\pi)} p_1(s')$ is even since $s, \lambda \models \s(p_1)$. Hence, we conclude that $\mathtt{Last}(\rho), \lambda[\rho] \models \s(p_1)$.

Now, consider $\as(p_2)$. By contradiction, assume that $\mathtt{Last}(\rho), \lambda[\rho] \not\models \as(p_2)$, i.e., that $q = \prob_{\mymdp, \mathtt{Last}(\rho)}^{\lambda[\rho]}[p_2] < 1$. It implies that $\prob_{\mymdp, s}^{\lambda}[p_2] \leq 1 - \prob_{\mymdp, s}^{\lambda}[\rho] \cdot (1 - q) < 1$ (since any prefix $\rho \in \mathtt{Pref}(\out_s^{\mymdp}(\lambda))$ has strictly positive probability). This contradicts the hypothesis that $s, \lambda \models \as(p_2)$. Hence it holds that $\mathtt{Last}(\rho), \lambda[\rho] \models \as(p_2)$ and we are done.
\end{proof}

\smallskip\noindent\textbf{Existence of UGECs.} The next lemma establishes that at least one UGEC must exist in~$\mymdp$ since $s_0 \models \s(p_1) \wedge \as(p_2)$.

\begin{lemma}
\label{lem:ugec}
The following holds: $s_0 \models \s(p_1) \wedge \as(p_2) \implies \ugec(\mymdp) \neq \emptyset$.
\end{lemma}

\begin{proof}
Given a set of plays $\Pi \subseteq \mathtt{Plays}(G)$, we define
\[
\mathtt{States}(\Pi) = \{s \in S \mid \exists\, \pi \in \Pi,\, \exists\, n \in \mathbb{N},\, \pi(n) = s\}.
\]
To prove this lemma, we first study the following set $\mathcal{S}$ of subsets of $S$:
\begin{align*}
\mathcal{S}=\left\{ R \subseteq S \mid \exists\, s \in S,\, \exists\, \lambda \in \Lambda,\, (s,\lambda \models \s(p_1) \wedge \as(p_2)) \wedge (R =  \mathtt{States}(\out_s^{\mymdp}(\lambda)))\right\}.
\end{align*}
Intuitively, this set contains any subset of $S$ that captures all states reachable by some witness strategy $\lambda$, from some state $s \in S$.
First note that $s_0 \models \s(p_1) \wedge \as(p_2)$ implies that $\mathcal{S}$ is non-empty, as for a witness strategy $\lambda$, $R =  \mathtt{States}(\out_{s_0}^{\mymdp}(\lambda)) \in \mathcal{S}$, by definition.

Second, we show that all minimal elements of $\mathcal{S}$ for set inclusion $\subseteq$ are UGECs, i.e., that for all $R\in\min_{\subseteq}(\mathcal{S})$, it holds that $R \in \ugec(\mymdp)$. This will establish our lemma. By Definition~\ref{def:ugec}, we have to prove that for each $R \in\min_{\subseteq}(\mathcal{S})$, the following properties hold:
\begin{enumerate}[label=\textbf{\alph*)}]
\item $R$ is an EC in $\mymdp$,
\item $\forall\, s\in R,\, s\models_{\mymdp_{\downharpoonright R}} \s(p_1) \wedge \as(\Diamond R^{\max}_{\even}(p_1))$,
\item $\forall\, s\in R,\, s\models_{\mymdp_{\downharpoonright R}} \as(p_1) \wedge \as (p_2)$.
\end{enumerate}

Before proving those three items, we claim that for all $R\in \min_{\subseteq}(\mathcal{S})$, and $s \in R$, there exists a strategy $\lambda_R$ such that $s,\lambda_R\models_{\mymdp_{\downharpoonright R}} \s(p_1) \wedge \as(p_2)$, i.e., $\lambda_R$ satisfies the property without leaving $R$. This is a direct consequence of Lemma~\ref{lem:suf-bwc} and the minimality of $R$ in $\mathcal{S}$ for the $\subseteq$ order. We use strategy $\lambda_R$ in the rest of the proof.

Item \textbf{a)}. We first prove that $R$ is strongly connected. By contradiction, assume it is not the case, i.e., that there exist $s, s' \in R$ such that there is no path in $R$ from $s$ to $s'$. Then, let $R'$ be the set of states reachable with strategy $\lambda_R$ from a prefix $\rho$ ending in $s$. By Lemma~\ref{lem:suf-bwc}, we have that $R' \in \mathcal{S}$. But as there is no path from $s$ to $s'$ in $R$, we have that $s' \not\in R'$ and $R' \subsetneq R$. This contradicts the minimality of $R$, hence we conclude that $R$ is strongly connected. We now have to prove that $R$ is a trap for $\playerTwo$. Clearly, for any state $s \in R \cap S_1$, we have $\mathtt{Succ}(s) \cap R \neq \emptyset$, as $R$ is strongly connected. Hence, it remains to show that for all $s \in R \cap S_2$, we have $\mathtt{Succ}(s) \subseteq R$. By contradiction, fix some $s \in R \cap S_2$ and assume there exists $s' \not\in R$ such that $(s, s') \in E$. As $R$ belongs to $\mathcal{S}$, recall that $R =\mathtt{States}(\out_{s''}^{\mymdp}(\lambda))$ for some strategy $\lambda$ and state $s''$. Since $s \in R$, there exists a prefix $\rho \in \mathtt{Pref}(\out_{s''}^{\mymdp}(\lambda))$ such that $\mathtt{Last(\rho)} = s$. But then, prefix $\rho' = \rho \cdot s'$ also belongs to $\mathtt{Pref}(\out_{s''}^{\mymdp}(\lambda))$, and $s' \in R$. Thus, we conclude that $R$ is indeed a trap for $\playerTwo$.

Item \textbf{b)}. Fix any $s \in R$ and let us prove that $s \models_{\mymdp_{\downharpoonright R}} \s(p_1) \wedge \as(\Diamond R^{\max}_{\even}(p_1))$. As seen above, from the minimality of $R$ and Lemma~\ref{lem:suf-bwc}, we know that $s,\lambda_R\models_{\mymdp_{\downharpoonright R}} \s(p_1) \wedge \as(p_2)$. Now, again from the minimality of $R$ in $\mathcal{S}$, we know that in the subtree induced by $\out_{s}^{\mymdp}(\lambda_R)$, any non-empty subset $R' \subseteq R$ is dense. That is, for every prefix $\rho$, the subtree defined by $\lambda_R$ from $\rho$ reaches a state of $R'$, and this holds in all subsequent subtrees. Using the reduction to a parity-B\"uchi game which underlies Theorem~\ref{thm:reach-prob-one}, we deduce from the density argument presented in Appendix~\ref{app:density} that we can build $\lambda$ from $\lambda_R$ such that $s, \lambda \models_{\mymdp_{\downharpoonright R}} \s(p_1) \wedge \as(\Diamond R')$.
It remains to argue that $R^{\max}_{\even}(p_1)$ is non-empty to prove this item. This is necessarily true, otherwise $R^{\max}_{\odd}(p_1)$ would be non-empty and $\lambda_R$ would not ensure $\s(p_1)$ in $R$ (as $R^{\max}_{\odd}(p_1)$ would be a dense subset).

Item \textbf{c)}. This is trivial as, for $s \in R$, $\lambda_R$ enforces $s,\lambda_R\models_{\mymdp_{\downharpoonright R}} \s(p_1) \wedge \as(p_2)$, a stronger property.
\end{proof}

\smallskip\noindent\textbf{Reaching UGECs.} Let us now collect in $\mathcal{U}_{\min}= {\displaystyle \cup_{R \in\min_{\subseteq}(\mathcal{S})} R}$ all states that belong to minimal sets $R$ of $\mathcal{S}$, and establish that $\mathcal{U}_{\min}$ can be reached almost-surely if it holds that $s_0 \models \s(p_1) \wedge \as(p_2)$.

\begin{lemma}
\label{lemma_rightleft}
The following holds: $s_0 \models \s(p_1) \wedge \as(p_2) \implies s_0 \models \s(p_1) \wedge \as(\Diamond \mathcal{U}_{\min})$.
\end{lemma}

\begin{proof}
Let $\lambda$ be a witness for $s_0 \models \s(p_1) \wedge \as(p_2)$. By Lemma~\ref{lem:ugec}, we know that $\ugec(\mymdp)$ is non-empty, and so is $\mathcal{U}_{\min}$. Furthermore, we claim that $\mathcal{U}_{\min}$ is dense in the tree induced by $\out_{s_0}^{\mymdp}(\lambda)$. Indeed, by Lemma~\ref{lem:suf-bwc}, after any prefix $\rho \in \mathtt{Pref}(\out_{s_0}^{\mymdp}(\lambda))$, the following property holds: $\mathtt{Last}(\rho), \lambda[\rho] \models \s(p_1) \wedge \as(p_2)$ (hence $\mathcal{U}_{\min}$ is reached, repeating the previous arguments). Since this holds for all consistent prefixes, we have that $\mathcal{U}_{\min}$ is indeed dense in the tree of $\lambda$. Hence, again using the density argument presented in Appendix~\ref{app:density}, we can build $\lambda'$ from $\lambda$ such that $s_0, \lambda' \models \s(p_1) \wedge \as(\Diamond \mathcal{U}_{\min})$.
\end{proof}

\subsection{Algorithm}
\label{subsec:algo}
Lemma~\ref{lemma_leftright} and Lemma~\ref{lemma_rightleft} prove the correctness of the reduction presented in Lemma~\ref{lemma_asp_to_asr}. It is the cornerstone of our algorithm deciding whether the property $s_0 \models \s(p_1) \wedge \as(p_2)$ holds.

\begin{theorem}
\label{theorem_asp}
Given an MDP $\mymdp = (G = (S,E), S_1, S_2, \delta)$, a state $s_0 \in S$, and two priority functions $p_i\colon S \rightarrow \{1, \dotso, d\}$, $i \in \{1, 2\}$, it can be decided in ${\sf NP} \cap {\sf coNP}$ if $s_0 \models \s(p_1) \wedge \as(p_2)$. If the answer is $\yes$, then there exists an infinite-memory witness strategy, and infinite memory is in general necessary. This decision problem is at least as hard as solving parity games.
\end{theorem}

\begin{proof}
The algorithm can be sketched as follows:
\begin{enumerate}
\item Compute the set $\max_{\subseteq}(\sgec(\mymdp))$ of maximal super-good ECs, using~\cite{DBLP:conf/concur/AlmagorKV16}. Those are the maximal ECs satisfying condition $\mathbf{(1_U)}$ in Definition~\ref{def:ugec}. There are only polynomially many of them, and their computation is in ${\sf NP} \cap {\sf coNP}$.
\item For each of them, check if condition $\mathbf{(2_U)}$ holds using Lemma~\ref{lemma_strat2}. This can be done in polynomial time, and if a super-good EC does not satisfy $\mathbf{(2_U)}$, then it is also the case of all its sub-ECs (as seen in the proof of Lemma~\ref{lemma_strat2}). Hence, we have that
\[
\mathcal{U} = \left\lbrace C \in \max_{\subseteq}(\sgec(\mymdp)) \mid C \text{ satisfies } \mathbf{(2_U)} \right\rbrace.
\]
\item Decide if $s_0 \models \s(p_1) \wedge \as(\Diamond \mathcal{U})$ using Theorem~\ref{thm:reach-prob-one}. This is in ${\sf NP} \cap {\sf coNP}$. If it holds, then answer $\yes$, otherwise answer $\no$.
\end{enumerate}

The correctness of this algorithm was established in Lemma~\ref{lemma_asp_to_asr}. It belongs to ${\sf P}^{{\sf NP} \cap {\sf coNP}} = {\sf NP} \cap {\sf coNP}$~\cite{Bra79}, and it trivially generalizes classical parity games (e.g., by taking $p_2\colon s \mapsto 0$ for all $s \in S$).

Finally, let us discuss strategies. A witness strategy $\lambda$ plays as follows: (i) it plays as the finite-memory strategy witness for $s_0 \models \s(p_1) \wedge \as(\Diamond \mathcal{U})$ given by Theorem~\ref{thm:reach-prob-one} until a UGEC $C$ is reached, (ii) then it switches to the infinite-memory strategy $\lambda_{C}$ described in Definition~\ref{def:strat_ugec}. It is clear that such a strategy is a witness for $s_0 \models \s(p_1) \wedge \as(p_2)$, as expected.

Infinite memory is required in general, as shown in the UGEC $C$ depicted in Figure~\ref{fig:full}: there exists no finite-memory witness strategy in $C$. Indeed, assume $\playerOne$ is restricted to a finite-memory strategy $\lambda$. To be able to ensure $p_1$ on the play in which $\playerTwo$ always goes to~$c$ from $b$, $\playerOne$ must visit $d$ infinitely often, and because of the finite memory of $\lambda$, he must do it after a bounded number of steps along which $a$ is not visited: say $n$ steps. Hence, the probability to do it will be bounded from below by a strictly positive constant, here $2^{-\frac{n}{2}}$ (the probability that $\playerTwo$ chooses $c$ for $\frac{n}{2}$ times in a row), all along a consistent play. Therefore, $\playerOne$ will almost-surely visit $d$ infinitely often, and $p_2$ will actually be satisfied with probability zero.
\end{proof}

\section{Parity with threshold probability under parity constraints}
\label{sec:threshp}

We now turn to the problem $s_0 \models \s(p_1) \wedge \p{\sim c}(p_2)$ for $\sim\, \in \{>, \geq\}$ and $c \in \mathbb{Q} \cap [0, 1)$.

\smallskip\noindent\textbf{Very-good end-components.} In addition to UGECs, we need the new notion of \textit{very-good end-component}.

\begin{definition}
\label{def:vgec}
An end-component $C$ of $\mymdp$ is \emph{very-good} ($\vgec$) if the following two properties hold:
\begin{itemize}
	\item $\mathbf{(1_V)}$ $\forall\, s\in C,\, s\models_{\mymdp} \s(p_1)$;
	\item $\mathbf{(2_V)}$ $\forall\, s\in C,\, s\models_{\mymdp_{\downharpoonright C}} \as(p_1) \wedge \as (p_2)$, or equivalently, $s\models_{\mymdp_{\downharpoonright C}} \as(p_1 \cap p_2)$.
\end{itemize}
We introduce the following notations:
  \begin{itemize}
  	\item $\vgec(\mymdp)$ is the set of all $\vgec$s of $\mymdp$, 
	\item $\mathcal{V} = {\displaystyle \cup_{V \in \vgec(\mymdp)} V}$ is the set of states that belong to a $\vgec$ in $\mymdp$.
  \end{itemize}
\end{definition}

Note that in condition $\mathbf{(1_V)}$, $\playerOne$ is allowed to leave $C$ to ensure $\s(p_1)$: this is in contrast to condition $\mathbf{(1_U)}$ for UGECs, in Definition~\ref{def:ugec}. On the contrary, condition $\mathbf{(2_V)}$ is exactly the same as $\mathbf{(2_U)}$.

From these definitions, it is trivial to see that any UGEC is also a VGEC, but the converse is false. Consider Figure~\ref{fig:v-not-u} ($\delta$ is the uniform distribution): $\{a,b,c\}$ is a VGEC. The strategy ensuring $\mathbf{(2_V)}$ from $a$ is to go to $b$, and the strategy ensuring $\mathbf{(1_V)}$ from $a$ is to go to $d$. As we will prove in Lemma~\ref{lem:suf2} and as in all VGECs, $\playerOne$ can ensure $a \models \s(p_1) \land \p{> 1- \varepsilon}(p_2)$ for any $\varepsilon > 0$.
Still, $\{a,b,c\}$ is not a UGEC: no strategy ensures $\s(p_1)$ on $\mymdp_{\downharpoonright \{a,b,c\}}$, as $\playerTwo$ can enforce the play $(ab)^\omega$ that has odd maximal priority. This illustrates why the notion of UGEC is too strong when reasoning about threshold probability instead of almost-sure satisfaction, hence why we need to introduce VGECs.

\begin{figure}[tbh]
\centering
\scalebox{0.8}{\begin{tikzpicture}[every node/.style={font=\small,inner sep=1pt}]
\draw (0,0) node[rond,rouge] (s0) {$1,1$};
\draw (-2,1) node[carre,rouge] (s1) {$1,1$};
\draw (-2,-1) node[carre,vert] (s2) {$2,2$};
\draw (2,0) node[carre,jaune] (s3) {$0,1$};
\draw (0,-0.6) node (l0) {$a$};
\draw (-2.5,1) node (l1) {$b$};
\draw (-2.5,-1) node (l2) {$c$};
\draw (2,-0.6) node (l3) {$d$};
\draw[-latex] (s0) to (s3);
\draw[-latex] (s1) to (s2);
\draw[-latex] (s2) to (s0);
\draw[-latex] (s0) to[out=100,in=355] (s1);
\draw[-latex] (s1) to[out=290,in=175] (s0);
\draw (s3) edge[-latex,out=330,in=30,looseness=4,distance=1cm] (s3);
\end{tikzpicture}}
\caption{The EC $\{a,b,c\}$ is very-good but not ultra-good, as $\playerOne$ has to leave it to ensure $\s(p_1)$.}
\label{fig:v-not-u}
\end{figure}
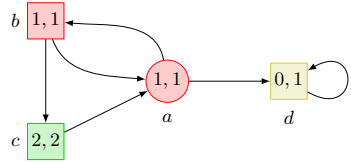

\smallskip\noindent\textbf{Available strategies in VGECs}. As for UGECs, we will use witness strategies for conditions $\mathbf{(1_V)}$ and $\mathbf{(2_V)}$. 
Deciding if condition $\mathbf{(1_V)}$ holds is solving a classical parity game, which lies in ${\sf NP} \cap {\sf coNP}$~\cite{DBLP:journals/ipl/Jurdzinski98}. Uniform pure memoryless witness strategies exist. We denote by $\lambda_{1}$ such a witness. For simplicity of presentation, \textit{we assume in the following that all states of $\mymdp$ satisfy $\mathbf{(1_V)}$}, as otherwise they will trivially not satisfy the properties we consider (as $\s(p_1)$ will not be ensured).
 For condition $\mathbf{(2_V)}$, we established in Lemma~\ref{lemma_strat2} that deciding if it holds can be done in polynomial time and that uniform randomized memoryless witness strategies exist. We denote by $\lambda_{2, C}$ such a witness.

\smallskip\noindent\textbf{Reaching VGECs.} We prove a strong relationship between the measure of paths that satisfy the two parity objectives $p_1$ and $p_2$, and the measure of paths that reach VGECs, under any strategy.

\begin{lemma}
\label{lem:prel-ls}
For all $s \in S$, and all $\lambda \in \Lambda$, the following holds:
\(
\prob^{\lambda}_{\mymdp, s}[\Diamond \mathcal{V}] \geq \prob^{\lambda}_{\mymdp, s}[p_1 \cap p_2].
\)
\end{lemma}

\begin{proof}
Let $\Pi =\{ \pi \in \out^{\mymdp}_s(\lambda) \mid \max_{s \in \infny(\pi)} p_1(s) \mbox{~and~} \max_{s \in \infny(\pi)} p_2(s) \mbox{~are even} \}$ be the set of consistent plays from $s$ that are winning for both $p_1$ and $p_2$. Let $q \in [0, 1]$ be the measure of $\Pi$ under $\lambda$. Now, let us define the set $\Pi'$ as $\Pi$ from which we remove the plays $\pi$ such that $\infny(\pi)$ is not an EC. By Lemma~\ref{lemma_as_ec}, the measure of $\Pi \setminus \Pi'$ is equal to zero, and so the measure of $\Pi'$ is also $q$. Now, let us note that any remaining play $\pi$ in $\Pi'$ visits (even infinitely many times) the EC $\infny(\pi)$, and by definition of $\Pi$, there exists a witness for $\mathbf{(2_V)}$: it suffices to play uniformly at random in $\infny(\pi)$. Since $\mathbf{(1_V)}$ is satisfied everywhere, we have that $\infny(\pi)$ is a VGEC. All the states of such VGECs belong to ${\cal V}$ by definition, and so we are done as we have proved that $\prob^{\lambda}_{\mymdp, s}[\Diamond \mathcal{V}] \geq q$.
\end{proof}

\smallskip\noindent\textbf{Limit-sure satisfaction in VGECs.} For each state in a VGEC, we claim that the parity objective $p_2$ can be satisfied with probability arbitrarily close to one, while ensuring $p_1$ surely.

\begin{lemma}
\label{lem:suf2}
Let $C \in \vgec(\mymdp)$. For all $s \in C$ and $\varepsilon \in (0, 1]$, the following property holds: $s \models \s(p_1) \land \p{> 1- \varepsilon}(p_2)$.
\end{lemma}

\begin{proof}
Our goal is to define a witness strategy $\lambda_{\varepsilon}$. First, for $\varepsilon \in (0,1]$, we fix an infinite sequence of strictly positive rational probabilities $f\colon \mathbb{N} \rightarrow \mathbb{Q} \cap (0,1]$ such that the product of all those probabilities is larger than $1-\varepsilon$, i.e., $\prod_{i \in \mathbb{N}} f(i) > 1- \varepsilon$. Such a sequence always exists. 
In turn, we use Lemma~\ref{lemma_strat2} to associate to $f$ a sequence of natural numbers $g\colon \mathbb{N} \rightarrow \mathbb{N}$ such that if $\lambda_{2,C}$ (the witness strategy for $\mathbf{(2_V)}$) is played for $g(i)$ steps from any $s' \in C$, then the set $D^{\max}_{\even}(p_1)$ associated to $\lambda_{2, C}$ is visited during those $g(i)$ steps with probability larger than $f(i)$.

Now, we are in position to define for all $\varepsilon \in (0,1]$, a strategy $\lambda_{\varepsilon, C}$ that enforces, from $s \in C$, the property $\s(p_1) \land \p{>1-\varepsilon}(p_2)$.
The strategy $\lambda_{\varepsilon, C}$ uses a counter $i$ whose value is initially equal to $0$. At round $i$, the strategy plays as $\lambda_{2,C}$ for $g(i)$ steps. During those last $g(i)$ steps, if the set $D^{\max}_{\even}(p_1)$ is visited, then the counter $i$ is incremented and the next round is executed. Otherwise, the strategy switches to $\lambda_{1}$ forever.

It is easy to see that, because of the definition that we have used for the sequence $g$, the probability that we ever switch to the strategy $\lambda_{1}$ is less than $\varepsilon$. So with probability larger than $1-\varepsilon$, we always play $\lambda_{2,C}$, which implies that $\lambda_{\varepsilon, C}$ ensures $\p{>1-\varepsilon}(p_2)$. Also, on all consistent plays in which strategy $\lambda_{2,C}$ is played forever, we know that the maximal priority seen infinitely often for $p_1$ is even (by definition of $D^{\max}_{\even}(p_1)$). On the other plays, as we play $\lambda_{1}$, we also have that the maximal priority seen infinitely often for $p_1$ is even. So, we can also conclude that $\lambda_{\varepsilon, C}$ ensures $\s(p_1)$, which concludes our proof.
\end{proof}

\subsection{The strict threshold case}

We first solve the strict threshold case with an approach that relies on VGECs. Similarly to what we did in Section~\ref{sec:asp}, we will establish a reduction of the decision problem for $s_0\models \s(p_1) \wedge \p{> c}(p_2)$ to a reachability problem toward the set $\mathcal{V}$, i.e., the union of VGECs. The actual algorithm will be detailed in Section~\ref{sec:threshold_algo}.

The first lemma gives a sufficient condition under which the property is satisfied. Note that its proof is constructive and tells us how to construct witness strategies. 

\begin{lemma}
\label{lemma:thres-suf-strict}
The following holds: $s_0 \models \s(p_1) \wedge \p{> c}(\Diamond \mathcal{V}) \implies s_0\models \s(p_1) \wedge \p{> c}(p_2)$.
\end{lemma}

\begin{proof}
Let $\lambda_{\Diamond {\cal V}}$ be a uniform pure memoryless strategy that ensures to reach ${\cal V}$ with probability $q > c$ from $s_0$ (it exists by Lemma~\ref{lemma_opti_reach}). Recall that $\lambda_{1}$ is a uniform pure memoryless strategy that enforces $\s(p_1)$ from any state in $\mymdp$. Let $\varepsilon$ be such that $q > \frac{c}{1-\varepsilon}$, and for all $C \in \vgec(\mymdp)$, let $\lambda_{\varepsilon,C}$ be a strategy that ensures $\s(p_1) \land \p{>1-\varepsilon}(p_2)$ when played from any state in $C$, as defined in the proof of Lemma~\ref{lem:suf2}.

We construct the strategy $\lambda$ that witnesses the desired property starting from the elements defined above. Let $r$ be a number of steps sufficient to ensure that ${\cal V}$ is reached with a probability larger than $\frac{c}{1-\varepsilon}$ when playing $\lambda_{\Diamond {\cal V}}$ from $s_0$: it exists by Lemma~\ref{lemma_opti_reach}. Strategy $\lambda$ plays as $\lambda_{\Diamond {\cal V}}$ for $r$ steps. If a VGEC $C$ of $\mymdp$ is reached during those $r$ steps, $\lambda$ immediately switches to $\lambda_{\varepsilon,C}$. Otherwise, after the $r$ steps, $\lambda$ switches to $\lambda_1$ forever. It is easy to verify that, when $\lambda$ is played, the probability that $p_2$ holds is larger than $\frac{c}{1-\varepsilon}\cdot (1-\varepsilon)$, and so it is larger than~$c$. Furthermore, $p_1$ holds on all consistent plays, which proves that $s_0, \lambda \models \s(p_1) \wedge \p{> c}(p_2)$ holds.
\end{proof}

This second lemma gives a necessary condition for the property to hold.

\begin{lemma}
\label{lemma:thres-suf}
The following holds: $s_0\models \s(p_1) \wedge \p{> c}(p_2) \implies s_0 \models \s(p_1) \wedge \p{> c}(\Diamond \mathcal{V})$.
\end{lemma}

\begin{proof}
First, recall that $\lambda_{1}$ is a uniform pure memoryless strategy that enforces $\s(p_1)$ from any state in $\mymdp$. Second, as $s_0 \models \s(p_1) \wedge \p{> c}(p_2)$, by Lemma~\ref{lem:prel-ls} and Lemma~\ref{lemma_opti_reach}, there exists a uniform pure memoryless strategy $\lambda_{\Diamond {\cal V}}$ that reaches $\mathcal{V}$ with probability $q > c$ from $s_0$.

We construct a strategy $\lambda$ that witnesses $s_0 \models \s(p_1) \wedge \p{> c}(\Diamond {\cal V})$ starting from $\lambda_{1}$ and $\lambda_{\Diamond {\cal V}}$ above. Let $r$ be a number of steps such that, if $\lambda_{\Diamond {\cal V}}$ is played for $r$ steps, then the probability to reach ${\cal V}$ from $s_0$ is larger than $c$: it exists by Lemma~\ref{lemma_opti_reach}. 
Strategy $\lambda$ starts from $s_0$ by playing as $\lambda_{\Diamond {\cal V}}$. It stops as soon as $\mathcal{V}$ is reached or if $r$ steps have been played and $\mathcal{V}$ has not been reached: in both cases, it switches to $\lambda_{1}$ forever. It is easy to check that $\lambda$ reaches $\mathcal{V}$ with probability larger than $c$ and it enforces $\s(p_1)$ by definition of $\lambda_1$ and prefix-independence of the parity objective. Hence, it holds that $s_0, \lambda \models \s(p_1) \wedge \p{> c}(\Diamond \mathcal{V})$.
\end{proof}

\subsection{The non-strict threshold case}

First, we note that, as we have solved the strict case above, the only interesting remaining case is when $\playerOne$, while surely forcing $p_1$, can force $p_2$ with probability $c$, but no more. The following two lemmas present a solution to this case. The main conceptual tool here is UGECs. As for the strict case, the corresponding algorithm will be detailed in Section~\ref{sec:threshold_algo}.

The first lemma gives a sufficient condition. Its proof is constructive: it explains how to build witness strategies. Recall that $\mathcal{U} = {\displaystyle \cup_{U \in \ugec(\mymdp)} U}$ is the set of states that belong to a $\ugec$ in $\mymdp$.

\begin{lemma}
\label{lemma:thres-suf-non-strict}
The following holds: $s_0 \models \s(p_1) \wedge \p{\geq c}(\Diamond \mathcal{U}) \implies s_0\models \s(p_1) \wedge \p{\geq c}(p_2)$.
\end{lemma}

\begin{proof}
Let $\lambda_{\Diamond {\cal U}}$ be a finite-memory strategy that witnesses property $s_0 \models \s(p_1) \wedge \p{\geq c}(\Diamond {\cal U})$. It exists by Theorem~\ref{theorem_reach_p}. For all $C \in \ugec(\mymdp)$,
let $\lambda_C$ be an infinite-memory witness for $s \models \s(p_1) \wedge \as(p_2)$, for all $s \in C$: it exists by Lemma~\ref{lem:suf}.
From those strategies, we define the strategy $\lambda$ as follows. From $s_0$, $\lambda$ plays as $\lambda_{\Diamond {\cal U}}$ up to reaching a UGEC $C$ of $\mymdp$, or forever if ${\cal U}$ is never reached. If some UGEC $C$ is reached, then it switches to $\lambda_C$ forever.

We claim that $\lambda$ is a witness for $s_0\models \s(p_1) \wedge \p{\geq c}(p_2)$. First, $\s(p_1)$ clearly holds as strategy $\lambda_{\Diamond \mathcal{U}}$ and all strategies $\lambda_C$ ensure it. Second, $\lambda_{\Diamond {\cal U}}$ ensures to reach ${\cal U}$ with probability at least~$c$, and when some UGEC $C \subseteq {\cal U}$ is reached, the strategy  $\lambda_{C}$ ensures $p_2$ with probability one. Thus, $\lambda$ ensures $p_2$ with probability at least $c$ from $s_0$, and we are done.
\end{proof}

We now turn to a lemma that gives a necessary condition, keeping in mind that we are interested in the case where $\playerOne$ cannot ensure probability strictly larger than $c$.

\begin{lemma}
\label{lemma:thres-nec-non-strict}
The following holds:
\[
(s_0\models \s(p_1) \wedge \p{\geq c}(p_2)) \wedge (s_0\not\models \s(p_1) \wedge \p{> c}(p_2)) \implies s_0 \models \s(p_1) \wedge \p{\geq c}(\Diamond \mathcal{U}).
\]
\end{lemma}

\begin{proof}
Let $\lambda$ be a strategy that witnesses $s_0\models \s(p_1) \wedge \p{\geq c}(p_2)$. As $s_0\not\models \s(p_1) \wedge \p{> c}(p_2)$, the measure of $\Pi=\{ \pi \in \out^{\mymdp}_{s_0}(\lambda) \mid \max_{s \in \infny(\pi)} p_1(s) \mbox{~and~} \max_{s \in \infny(\pi)} p_2(s) \mbox{~are even}  \}$ is exactly equal to $c$. We define $\Pi'$ as $\Pi$ from which we remove all play $\pi$ such that $\infny(\pi)$ is not an EC. By Lemma~\ref{lemma_as_ec}, the measure of $\Pi'$ is equal to the one of $\Pi$. As all plays in $\Pi'$ satisfy both $p_1$ and $p_2$ and are such that $\infny(\pi)$ is an EC, we conclude that they all reach a VGEC (as conditions $\mathbf{(1_V)}$ and $\mathbf{(2_V)}$ are satisfied in $\infny(\pi)$).

Let $\pi \in \Pi'$ and let $\rho$ be a prefix of $\pi$ such that $\mathtt{Last}(\rho) \in C$ for some $C \in \vgec(\mymdp)$. We claim that $\mathtt{Last}(\rho), \lambda[\rho] \models \s(p_1) \wedge \as(p_2)$. We prove it by contradiction. Assume that $\lambda[\rho]$ only enforces $p_2$ with probability $1 - \varepsilon$ for some $\varepsilon > 0$. As $\mathtt{Last}(\rho)$ belongs to the VGEC~$C$, we invoke Lemma~\ref{lem:suf2} to build a strategy $\lambda_{\rho}$ such that $\mathtt{Last}(\rho), \lambda_{\rho} \models \s(p_1) \wedge \p{> 1 - \varepsilon'}(p_2)$ for $\varepsilon' < \varepsilon$. Now, to obtain the contradiction, we construct a new strategy $\lambda'$ from $s_0$ that plays as $\lambda$ but switches to $\lambda_{\rho}$ after prefix $\rho$. We have that $s_0, \lambda' \models \s(p_1) \wedge \p{> c}(p_2)$, which contradicts our hypothesis. Hence, we conclude that $\mathtt{Last}(\rho), \lambda[\rho] \models \s(p_1) \wedge \as(p_2)$ holds.

Finally, applying Lemma~\ref{lemma_rightleft} to $\mathtt{Last}(\rho)$, we know that $\mathtt{Last}(\rho) \models \s(p_1) \wedge \as(\Diamond \mathcal{U})$. As the measure of $\Pi'$ is equal to $c$ and this reasoning holds for any $\pi \in \Pi'$, we conclude that $s_0, \lambda \models \s(p_1) \wedge \p{\geq c}(\Diamond \mathcal{U})$ and we are done.
\end{proof}

\subsection{Algorithm}
\label{sec:threshold_algo}

Based on the reductions shown above, we can now establish an algorithm and complexity results for the threshold problem.

\begin{theorem}
\label{theorem_thresh_p}
Given an MDP $\mymdp = (G = (S,E), S_1, S_2, \delta)$, a state $s_0 \in S$, and two priority functions $p_i\colon S \rightarrow \{1, \dotso, d\}$, $i \in \{1, 2\}$, it can be decided in ${\sf NP} \cap {\sf coNP}$ if $s_0 \models \s(p_1) \wedge \p{\sim c}(p_2)$ for $\sim\, \in \{>, \geq\}$ and $c \in \mathbb{Q} \cap [0, 1)$. If the answer is $\yes$, then there exists an infinite-memory witness strategy, and infinite memory is in general necessary. This decision problem is at least as hard as solving parity games.
\end{theorem}

\begin{proof}
The algorithm can be sketched as follows:
\begin{enumerate}
\item Remove from $\mymdp$ all states where $\s(p_1)$ does not hold, as well as their attractor for $\playerTwo$: if $s_0$ is removed, then answer $\no$. Let $\mymdp'$ be the remaining MDP. This operation is in ${\sf NP} \cap {\sf coNP}$ as it consists in solving a classical parity game~\cite{DBLP:journals/ipl/Jurdzinski98}.
\item Compute the set $\mathcal{V}$ representing the union of VGECs in $\mymdp'$. This can be done in polynomial time by computing the maximal ECs of $\mymdp'$ and applying Lemma~\ref{lemma_strat2} to check condition $\mathbf{(2_V)}$ for each of them (condition $\mathbf{(1_V)}$ holds thanks to the previous step).
\item Decide if $s_0 \models \s(p_1) \wedge \p{> c}(\Diamond \mathcal{V})$ using Theorem~\ref{theorem_reach_p}. This is in ${\sf NP} \cap {\sf coNP}$. If it holds, then answer $\yes$. If it does not hold and $\sim$ is $>$, then answer $\no$, otherwise, i.e., if $\sim$ is $\geq$, continue with the next step.
\item Use the sub-algorithm described in Theorem~\ref{theorem_asp} to compute the set $\mathcal{U}$ representing the union of UGECs in $\mymdp'$. This is in ${\sf NP} \cap {\sf coNP}$.
\item Decide if $s_0 \models \s(p_1) \wedge \p{\geq c}(\Diamond \mathcal{U})$ using Theorem~\ref{theorem_reach_p}. This is in ${\sf NP} \cap {\sf coNP}$. If it holds, answer $\yes$, otherwise answer $\no$.
\end{enumerate}

The correctness of this algorithm follows from Lemma~\ref{lemma:thres-suf-strict}, Lemma~\ref{lemma:thres-suf}, Lemma~\ref{lemma:thres-suf-non-strict}, and Lemma~\ref{lemma:thres-nec-non-strict}. It belongs to ${\sf P}^{{\sf NP} \cap {\sf coNP}} = {\sf NP} \cap {\sf coNP}$~\cite{Bra79}, and it trivially generalizes classical parity games (e.g., by taking $p_2\colon s \mapsto 0$ for all $s \in S$).

Finally, let us discuss strategies. Witness strategies for the case $>$ (resp.~$\geq$) were described in Lemma~\ref{lemma:thres-suf-strict} (resp.~Lemma~\ref{lemma:thres-suf-non-strict}). In both cases, infinite memory is in general required, because it is in general necessary to play optimally in both VGECs and UGECs. For UGECs, see Theorem~\ref{theorem_asp} for an example. For VGECs, consider the VGEC $\{a, b, c\}$ in the MDP of Figure~\ref{fig:v-not-u}. We claim that for every finite-memory strategy $\lambda$ ensuring $\s(p_1)$, the probability to ensure $p_2$ is zero, hence there is no finite-memory witness for $a \models \s(p_1) \wedge \p{>1-\varepsilon}(p_2)$. As argued for the UGEC case, in order to ensure $p_1$ on the play in which $\playerTwo$ always goes to $a$ from $b$, $\playerOne$ must go to $d$ at some point, and because of the finite memory of $\lambda$, he must do it after a bounded number of steps along which $c$ is not visited: say $n$ steps. Again, the probability to do it will be bounded from below by a strictly positive constant, here $2^{-\frac{n}{2}}$ (the probability that $\playerTwo$ chooses $a$ for $\frac{n}{2}$ times in a row), all along a consistent play. Therefore, $\playerOne$ will almost-surely go to $d$, and $p_2$ will actually be satisfied with probability zero.
\end{proof}

\section{Conclusion}
We further extended the beyond worst-case synthesis framework by providing tools to reason about $\omega$-regular conditions, representing functional requirements of systems. We studied the case of two parity objectives and proved ${\sf NP} \cap {\sf coNP}$ membership for all considered variants.

It is interesting to note that our algorithms can easily be generalized to more than two parity objectives as long as we consider only the $\s$ and $\as$ operators. Indeed, we have that for any MDP $\mymdp$, any state $s$ in $\mymdp$, and any number of priority functions $p_1, \dotso p_n$, it holds that $s \models \bigwedge_{i} \s(p_i) \bigwedge_{j} \as(p_j) \iff s \models \s\big(\bigwedge_{i} p_i\big) \wedge \as\big(\bigwedge_{j} p_j\big)$, and it is easy to reduce the latter problem to $s' \models \s(p') \wedge \as(p'')$ on a (larger) MDP $\mymdp'$, using classical techniques (e.g., any conjunction of parity objectives can be expressed as a Muller condition~\cite{DBLP:conf/fossacs/ChatterjeeHP07}, that in turn can be transformed into a single parity condition on a larger graph~\cite{DBLP:conf/fsttcs/Loding99}). Extending this generalization to the operator $\p{\sim c}$ is more challenging and would require to mix our techniques to methods for percentile queries~\cite{percentile2017}: an interesting direction for future work.

Another question is the limits of finite-memory strategies. We saw that in general, infinite memory is needed to satisfy problems involving two parity objectives. We would like to investigate under which additional conditions finite-memory strategies suffice, and to develop corresponding algorithms, as finite-memory strategies are of great practical interest.

\bibliography{biblio}

\newpage
\appendix
\section{Additional tools related to density}
\label{app:density}

In this section, we develop a density argument that is needed to prove the right-to-left implications of Lemma~\ref{lemma_geq_to_as}, and Lemma~\ref{lemma_asp_to_asr} in Section~\ref{subsec:necessary}. We place this argument in appendix, as it essentially gives an explicit presentation of results that implicitly follow from~\cite{DBLP:conf/concur/AlmagorKV16}: hence it is not a new contribution.

\smallskip\noindent\textbf{Density.} Let $\mymdp = (G = (S, E), S_1, S_2, \delta)$ be an MDP, $s\in S$ an initial state, $\lambda$ a strategy, and $R\subseteq S$. We say that $R$ is \textit{dense} in $\lambda$ from $s$ if and only if for all $\rho \in \mathtt{Pref}(\out_s^{\mymdp}(\lambda))$, there exists $\rho'$ such that $\rho \cdot \rho' \in \mathtt{Pref}(\out_s^{\mymdp}(\lambda))$ and $\mathtt{Last}(\rho') \in R$. That is, after all prefixes in the tree $\out_s^{\mymdp}(\lambda)$, there is a continuation that visits $R$.

\smallskip\noindent\textbf{B\"uchi objective.} The B\"uchi objective is a parity objective with priorities in $\{1, 2\}$, where states with priority $2$ are called B\"uchi accepting states, defining a set $B$: $\mathtt{B\"uchi}(B) = \{\pi\in\mathtt{Plays}(G) \mid \infny(\pi) \cap B \neq \emptyset\}$.

\smallskip\noindent\textbf{Almost-sure reachability under parity constraints.} We first give explicitly the construction of~\cite[Lemma 3]{almagor2016minimizing} (extended version of~\cite{DBLP:conf/concur/AlmagorKV16}), that we used implicitly in Theorem~\ref{thm:reach-prob-one}.

\begin{theorem}[\cite{DBLP:conf/concur/AlmagorKV16,almagor2016minimizing}]
\label{thm:reach-Almagor}
Given an MDP $\mymdp = (G = (S, E), S_1, S_2, \delta)$, a state $s_0 \in S$, a priority function $p\colon S \rightarrow \{1,2, \dots, d\}$, and a target set $R \subseteq S$, it can be decided in ${\sf NP} \cap {\sf coNP}$ if $s_0 \models \s(\neg (\Diamond R)\rightarrow p) \wedge \as(\Diamond R)$. If the answer is $\yes$, then there exists a finite-memory witness strategy. This decision problem is at least as hard as parity games.
\end{theorem}

To solve the decision problem of Theorem~\ref{thm:reach-Almagor}, Almagor et al.~construct a two-player zero-sum game with a conjunction of a B\"uchi and a parity objectives. We recall the construction here.
To simplify the formal definition and w.l.o.g., we make the hypothesis that in $\mymdp$, $E \subseteq S_1 \times S_2 \cup S_2 \times S_1$, i.e., states of $\playerOne$ and $\playerTwo$ alternate. Then the game $G^{\mymdp}_{R,p}$ is defined as follows.
  \begin{itemize}
  	\item The state space of $G^{\mymdp}_{R,p}$ is a copy of the state space of $\mymdp$ where the states in $S_2 \setminus R$ have been duplicated in three copies as shown in Figure~\ref{fig:ggt}. Then $S'$ is defined as $S'=S_1 \cup S_2 \cup (S_2 \setminus R) \times \{ {\tt square}, {\tt circle} \}$: $\playerOne$ owns the states in $S_1$ and $(S_2 \setminus R) \times \{ {\tt circle}\}$, and $\playerTwo$ owns the states in $S_2$ and $(S_2 \setminus R) \times \{ {\tt square}\}$.
	\item The set of edges are defined according to Figure~\ref{fig:ggt} with the additional property that states in $R$ are made absorbing. That is, $E'$ is the union of the following sets:
	\begin{itemize}
		\item $\{ (s,s') \mid (s,s') \in E \cap ((S_1\setminus R) \times S_2) \}$,
		\item $\{ (s,(s,\star)) \mid s \in S_2 \setminus R,\, \star \in \{ {\tt square},{\tt circle} \} ) \}$,
		\item $\{ ((s,\star),s') \mid (s,s') \in E \cap ((S_2\setminus R) \times S_1),\, \star \in \{ {\tt square},{\tt circle} \} \}$,
		\item $\{ (s,s) \mid s \in R \}$.
	\end{itemize}

	\item The priority function $p'$ is defined as $p'(s)=p(s)$ for all $s \in S_1 \cup S_2 \setminus R$, $p'((s,\star))=0$ for all $s \in S_2 \setminus R$ and $\star \in \{ {\tt square},{\tt circle} \}$, and $p'(s)=0$ for all $s \in R$.

	\item The set of B\"uchi states is $B=R \cup \{ (s,{\tt square}) \mid s \in S_2\setminus R \}$.
\end{itemize}

\begin{figure}[tbh]
\centering
\scalebox{0.8}{\begin{tikzpicture}[every node/.style={font=\small,inner sep=1pt}]
\draw (0,0) node[carre,bleu] (s0) {$i$};
\draw (6,0) node[carre,bleu] (s1) {$i$};
\draw (8,1) node[carre,double,thick,jaune] (s2) {$0$};
\draw (8,-1) node[rond,thick,jaune] (s3) {$0$};
\draw (4,0) node (l) {{\LARGE $\Longrightarrow$}};
\draw[-latex] (s0) to (1,0.5);
\draw[-latex] (s0) to (1,0);
\draw[-latex] (s0) to (1,-0.5);
\draw[-latex] (s2) to (9,1.5);
\draw[-latex] (s2) to (9,1);
\draw[-latex] (s2) to (9,0.5);
\draw[-latex] (s3) to (9,-1.5);
\draw[-latex] (s3) to (9,-1);
\draw[-latex] (s3) to (9,-0.5);
\draw[-latex] (s1) to (s2);
\draw[-latex] (s1) to (s3);
\end{tikzpicture}}
\caption{Replacing a $\player{2}$ state with priority $i$ by three new states, including a B\"uchi one.}
\label{fig:ggt}
\end{figure}
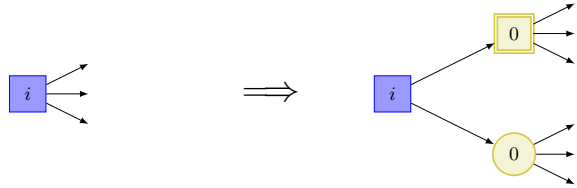

Careful inspection of~\cite{DBLP:conf/concur/AlmagorKV16,almagor2016minimizing} reveals that $s_0 \models \s(\neg (\Diamond R)\rightarrow p) \wedge \as(\Diamond R)$ if and only if $\playerOne$ has a winning strategy in the B\"uchi parity game $G^{\mymdp}_{R,p}$ from state $s_0$, and that finite-memory strategies suffice on both sides. Indeed, finite-memory strategies suffice for B\"uchi parity games (as they are a particular case of generalized parity games~\cite{DBLP:conf/fossacs/ChatterjeeHP07}), and from a finite-memory strategy in the B\"uchi parity game, one obtains straightforwardly a witness strategy for property $\s(\neg (\Diamond R)\rightarrow p) \wedge \as(\Diamond R)$. The proof presented in~\cite{DBLP:conf/concur/AlmagorKV16,almagor2016minimizing} further reduces the problem of deciding the winner in the B\"uchi parity game to the same problem on a mean-payoff parity game~\cite{DBLP:conf/lics/ChatterjeeHJ05}, in order to benefit from the ${\sf NP} \cap {\sf coNP}$ complexity of the latter. This reduction does hold since we can restrict ourselves to finite-memory strategies without loss of generality.

We now prove a related result that uses the notion of density of a set of states in a strategy tree.

\begin{lemma}
\label{lem:bp}
Given an MDP $\mymdp = (G = (S, E), S_1, S_2, \delta)$, a state $s \in S$, a priority function $p\colon S \rightarrow \{1,2, \dots, d\}$, a set $R \subseteq S$, and a strategy $\lambda$ such that $s, \lambda \models \s(p_1)$ and $R$ is dense in $\lambda$ from $s$, there is a winning strategy from $s$ in the associated B\"uchi parity game $G^{\mymdp}_{R,p}$ from state $s$.
\end{lemma}

\begin{proof}
We construct a strategy $\lambda'$ for $\playerOne$ to play in $G^{\mymdp}_{R,p}$ starting from the strategy $\lambda$, and we show that $\lambda'$ is winning for the B\"uchi parity objective.

Let $\alpha\colon S' \rightarrow S$ be a mapping from the states in $G^{\mymdp}_{R,p}$ to states in $\mymdp$ such that $\alpha(s)=s$ for all $s \in S_1 \cup S_2$ and $\alpha(s,\star)=\epsilon$ for all $(s,\star) \in (S_2\setminus R) \times \{ {\tt square}, {\tt circle} \}$. Then we extend $\alpha$ to map prefixes in $G^{\mymdp}_{R,p}$ to prefixes in $\mymdp$.

Now, we label each prefix $\rho \in \mathtt{Pref}(\out_s^{\mymdp}(\lambda))$ such that $\mathtt{Last}(\rho) \in S_2$ by a finite continuation $\mathtt{lab}(\rho)=\rho'$ such that (i) $\rho \cdot \rho' \in \mathtt{Pref}(\out_s^{\mymdp}(\lambda))$ and $\mathtt{last}(\rho') \in R$; and (ii) if $\rho'=\rho'_1 \cdot s \cdot \rho'_2$ is the label of $\rho$ and $s \in S_2$, then $\mathtt{lab}(\rho \cdot \rho'_1 \cdot s)=\rho'_2$. 	It should be clear that due to the density of $R$, such a labelling is always possible.

Now, we define the strategy $\lambda'$ from $\lambda$ and this labelling:
  \begin{itemize}
	\item For a prefix $\rho$ in the game such that $\mathtt{Last}(\rho) \in R$, the only possible choice is $\mathtt{Last}(\rho)$ as $R$ is absorbing in the game.
	\item For a prefix $\rho$ such that $\mathtt{Last}(\rho) \in S_1 \cup (S_2\setminus R) \times \{{\tt circle}\}$ and such that $\alpha(\rho)$ is not consistent with $\lambda$, we choose any $s'$ such that $(\mathtt{Last}(\rho),s')\in E'$. 
	\item For a prefix $\rho$ such that $\mathtt{Last}(\rho) \in S_1$ and such that $\alpha(\rho)$ is consistent with $\lambda$, we define $\lambda'(\rho)=\lambda(\alpha(\rho))$
	\item For a prefix $\rho$ such that $\mathtt{Last}(\rho) \in (S_2\setminus R) \times \{ {\tt circle} \}$ and such that $\alpha(\rho)$ is consistent with $\lambda$, then $\alpha(\rho)$ in the tree $\out_s^{\mymdp}(\lambda)$ is labelled with a finite path $\rho'$ that leads to a state in $R$, i.e., $\mathtt{lab}(\alpha(\rho))=\rho'$, then we define $\lambda'(\rho)=\mathtt{First}(\rho')$.
\end{itemize}	
	
Let us now show that $\lambda'$ wins the B\"uchi $B$ and parity $p'$ objectives in the game $G^{\mymdp}_{R,p}$ from state $s$. First, as $\lambda$ is enforcing $p$ in $\mymdp$, $p'$ replicates the priorities given by $p$, and absorbing states in $R$ have even priorities for $p'$, it is clear that $\lambda'$ enforces $p'$ in the game $G^{\mymdp}_{R,p}$. Now, for the B\"uchi objective, we consider the following case study on the plays $\pi$ consistent with $\lambda'$ in the game.
	\begin{itemize}
		\item If $\pi$ ends up in $R$, then it is clearly B\"uchi accepting.
		\item If $\pi$ is such that after a finite prefix $\rho$, $\playerTwo$ always chooses from states $s' \in S_2$ the copy $(s',{\tt circle})$, then according to the definition of $\lambda'$, the play will follow a finite path $\rho'$ to a state in $R$, and so $\pi$ reaches $R$ and as a consequence $\pi$ is B\"uchi accepting.
		\item Finally, if $\pi$ is such that it is never the case that $\playerTwo$ always chooses from a state $s'$ the copy $(s',{\tt circle})$, then $\playerTwo$ chooses infinitely often the copy $(s',{\tt square})$, and because $(s',{\tt square}) \in B$, we have that $\pi$ is also B\"uchi accepting.
	\end{itemize}
So in all cases, play $\pi$ satisfies the B\"uchi objective, and we are done.
\end{proof}

\smallskip\noindent\textbf{Consequence.} Lemma~\ref{lem:bp}, in combination with Theorem~\ref{thm:reach-Almagor} and the reduction of Theorem~\ref{thm:reach-prob-one}, has the interesting consequence that if $\lambda$ is such that $s, \lambda \models \s(p_1)$ and $R$ is dense in $\lambda$ from~$s$, then we can build $\lambda'$ such that $s, \lambda' \models \s(p_1) \wedge \as(\Diamond R)$. This is the property that we use in Lemma~\ref{lemma_geq_to_as}, but also in Lemma~\ref{lem:ugec} and Lemma~\ref{lemma_rightleft} to prove the right-to-left implication of Lemma~\ref{lemma_asp_to_asr}.
\end{document}